		\def\version{}				                                   %

\documentclass[12  pt,]{article}
\usepackage{cite}
\usepackage{subfig}
\usepackage{graphicx} 
\usepackage{url}

\usepackage{amsmath,amsthm,amssymb,a4,graphics,color}
\usepackage[mathscr]{eucal} 
\usepackage[active]{srcltx}
\usepackage[usenames,dvipsnames]{pstricks}
\usepackage{epsfig}
\usepackage{pst-grad}
\usepackage{pst-plot} 
\usepackage{textcomp}
\usepackage{cite}
\usepackage{natbib}
\usepackage{amsfonts}

\theoremstyle{definition}

\newtheorem{theorem}{Theorem}

\newtheorem{lemma}[theorem]{Lemma}

\newfam\Bbbfam 
\font\tenBbb=msbm10 
\font\sevenBbb=msbm7 
\font\fiveBbb=msbm5 
\textfont\Bbbfam=\tenBbb 
\scriptfont\Bbbfam=\sevenBbb 
\scriptscriptfont\Bbbfam=\fiveBbb

\DeclareMathOperator{\EX}{\mathbb{E}}
\DeclareMathOperator{\prob}{\mathbb{P}}

\begin{document}
\title{A Reduced-bias weighted least squares estimation of the Extreme Value Index }

\vspace{0.2cm}

\maketitle
\centerline
{\sc
	E. Ocran \footnote{Department of Statistics and Actuarial Science} $\footnote{eocran003@st.ug.edu.gh}$, R.  Minkah  \footnotemark[1]  $\footnote{rminkah@ug.edu.gh}$   and  K.  Doku-Amponsah \footnotemark[1]  $\footnote{kdoku@ug.edu.gh,\,\, kdoku-amponsah@ug.edu.gh}$    
}

\vspace{0.5cm}
\centerline{\sc University of Ghana}

\vspace{-1.95cm}

\renewcommand{\thefootnote}{}

\medskip

\vspace{0.2cm}

\begin{center}\small{\version}\end{center}

\bigskip

\begin{abstract}
	In this paper, we propose a reduced-bias estimator of the EVI for Pareto-type tails (heavy-tailed) distributions. This is derived using the weighted least squares method. It is shown that the estimator is asymptotically unbiased, asymptotically consistent and asymptotically normal under the second-order conditions on the underlying distribution of the data. The finite sample properties of the proposed estimator are studied through a simulation study. The results show that it is competitive to the existing estimators of the extreme value index in terms of bias and Mean Square Error. In addition, it yields estimates of $\gamma>0$ that are less sensitive to the number of top-order statistics, and hence, can be used for selecting an optimal tail fraction. The proposed estimator is further illustrated using practical datasets from pedochemical and insurance. 
	
	\vspace{0.5em}
\end{abstract}

\section{\textbf{Introduction}}
Statistics of extremes deals with the estimations of the occurrences of rare events. Such events include high quantile, exceedance probability and return periods. The knowledge of the frequency and magnitude of these events help in planning to mitigate the effects of their occurrences. The primary parameter in estimating these events is the extreme value index (EVI), $\gamma$ (or also known as the tail index). The extreme value index, either implicitly or explicitly, plays a vital role in estimating other extreme events \citep{bladt2021trimmed,elvidge2018using,kithinji2021adjusted}.  Hence in statistics of extremes, the primary focus is the estimation of the EVI, $\gamma$. The EVI is crucial in extreme value analysis because all inference requires that the EVI is known.  The EVI regulates the tail heaviness or flatness of the extreme value distribution, assuming the max-domain of attraction condition holds for the underlying distribution, $F$. We may have $\gamma \in \mathbb{R}$ and the heaviness of the tail function $\bar{F}:=1-F$, increases with $\gamma$. However, in this paper, we shall consider the particular case where $\gamma>0$. That is, we shall work in the domain of attraction for maxima of an extreme value distribution function (df),
\begin{equation}\label{E1}
H_{\gamma}(x)=exp\left(-\left(1+\gamma x\right)^{-1/\gamma}\right), \hspace{6pt}1+\gamma x\ge 0. 
\end{equation}

This class of distributions has been found to be useful in diverse fields, including but not limited to insurance \citep{minkah2021robust, minkah2020tail,rohrbeck2018extreme}, finance \citep{longin2016extreme,kithinji2021adjusted,gkillas2018application}, economics \citep{embrechts2003extremes,marimoutou2009extreme,muela2017application}, telecommunication \citep{mehrnia2021wireless,finkenstadt2003extreme} and Climatology \citep{onwuegbuche2019application,yozgatligil2018extreme}.
For $\gamma >0$, the underlying distribution, $F$ belongs to the Pareto-type distributions with the upper tail function expressed as
\begin{equation}\label{E2}
\bar{F}=x^{-1/\gamma}\ell_{F}(x); \hspace{6pt} x>1,
\end{equation}
or has a tail quantile function defined by
\begin{equation}\label{E4}
U(x)=x^{\gamma}\ell_{u}(x) x>1,
\end{equation}

where $\ell_{F}$ and $\ell_{u}$ are slowly varying functions at infinity, that is,
\begin{equation}\label{E3}
\dfrac{\ell_{\mathcal{B}}(tx)}{\ell_{\mathcal{B}}(x)}\to 1, \text{as} \hspace{6pt} x\to \infty, \text{for all} \hspace{6pt} t>0, \mathcal{B}\in\{F,u\}.
\end{equation}

Using the concept of the \textit{de Bruyn conjugate}, it can be shown that $\ell_{F}$ and $\ell_{u}$ are linked \citep{Beirlant2004}.

Let $\chi=\{X_{j}\}_{j=1}^{n}$ denote a sequence of independent and identically distributed random variables with distribution function,$F$, belonging to the class of distributions defined by (\ref{E2}). Again, let $X_{j,n}$ denote the $j^{th}$ order statistic associated with the  sample $\chi$. Define the weighted log-spacings
\begin{equation}\label{E5}
Z_{j}=j\log\left(\dfrac{X_{n-j+1,n}}{X_{n-j,n}}\right), 1\le j \le k.
\end{equation}

\citet{feuerverger1999} and \citet{Beirlant1999} demonstrated that the effect of bias can be accommodated by modelling the weighted log-spacing of the order statistics and also proposed the usage of a regression model on $Z_{j}$ with $\left(\dfrac{j}{k+1}\right)^{-\rho}$ as the explanatory variable. It can be seen from \citet{Beirlant1999} and \citet{minkah2021robust} that the $Z_{j}'s$ can be approximated with a regression model with exponential responses:
\begin{equation}\label{E7}
Z_{j}=\left(\gamma+b_{n,k}\left(\dfrac{j}{k+1}\right)^{-\rho}\right)f_{j}, 1\le j \le k,
\end{equation}
where   $b_{n.k}=b((n+1)/(k+1)) \to 0$ as $k, n\to \infty$, $\rho<0$ is a second-order parameter, and  $f_{j}'s$ follow exponential distribution with a unit mean. The authors stated that this representation permits the generation of reduced-bias estimates of $\gamma$. In the case of strict Pareto distributions, (\ref{E7}) reduces to
\begin{equation}\label{E8}
Z_{j}=\gamma f_{j}, j=1,...,k,
\end{equation}
and the estimator of $\gamma$ in this case is the usual Hill estimator \citep{hill1975}.

Different estimation techniques have been proposed for estimating $\gamma$ (see for example, \citealt{ Pickands1975,hill1975,Dekkers1989moment,Beirlant1996,Beirlant2005,Caeiro2005a,Buitendag2018}). The Hill estimator \citep{hill1975} till date plays a vital role in many applications and it is the most widely used estimator in practice. Using the Renyi's representation, the \citet{hill1975} estimator can be expressed as the mean of the weighted log-spacings, 
\begin{equation}\label{E9}
\gamma_{k}^{H}=\dfrac{1}{k}\sum_{j=1}^{k}Z_{j},
\end{equation}

A desirable property of the Hill estimator is that it has a minimal variance and is asymptotically normal \citep{de1998asymptotic,mason1982laws}. A drawback which makes the usage of the Hill estimator in application difficult is its sensitivity to the number of top order statistics, $k$. This drawback is common to most semi-parametric estimators currently in existence. Any estimator which is sensitive to the value of $k$ may be challenging to use in application since a little change in the value of $k$ may result in large change in the value of the extreme value index. In addition it makes it challenging to select the optimal $k.$

Since the Hill estimator is obtained with $\ell_{F}(x)=1$ in (\ref{E2}), in order to obtain a reduced-bias estimator, a functional form of $\ell_{F}$ is needed. Therefore \citet{Beirlant1999} assume that for the Hall class of distributions \citep{hall1982} the slowly varying part $\ell_{F}$ satisfies a second-order condition :\\
\textbf{\textit{Assumption ($R_{\ell}$)}:} \textit{There exists some auxiliary function $b>0$ that is regularly varying with index $\rho<0$ satisfying $b(x) \to 0$ as $x\to \infty$, such that for all $\lambda\ge 1,$ as $x \to \infty$ }
\begin{equation}\label{E10}
\log \dfrac{\ell_{F}(\lambda x)}{\ell_{F}(x)}\approx b(x)k_{\rho}(\lambda)
\end{equation}
\textit{with $k_{p}(\lambda)=\dfrac{\lambda^{\rho}-1}{\rho}$}.
That is, under assumption $R_{\ell}$, we assume the slowly varying part has a structure so we do not ignore it and hence reduces the bias accompanying the approximation of $\ell_{F}$ as a constant. 

\citet{Beirlant2002} further approximated (\ref{E7}) under assumption $R_{\ell}$ as
\begin{equation}\label{E11}
Z_{j}=Z_{j}(\rho,\epsilon_j)=\gamma+b_{n,k}C_{j}(\rho)+\epsilon_{j},
\end{equation}
where $b_{n,k}=b\left(n/k\right)\to 0$  ($k, n\to \infty$) is the slope,  $C_j=C_{j}(\rho)=\left(j/(k+1)\right)^{-\rho}$ is the covariate, $\gamma$ is the intercept of the regression line,and the error term, $\epsilon_{j}$ is asymptotically normal with variance, $\gamma^2$, and mean, $0$. The regression model in (\ref{E11}) has independent exponential responses \citep{Beirlant1999}. Therefore, for fixed $j$, the $Z_j$'s are approximately independent exponentially distributed with mean 

\begin{equation}\label{ME}
\mu_{j}=\EX(Z_j)=\gamma+\EX\big[\hat b_{n,k}(\hat\rho)C_j(\hat\rho)\big],
\end{equation}
where   $\hat \rho$  is  estimator  of  $\rho$   and  $\hat b_{n,k}(\hat\rho)$  is  estimator  of  $ b_{n,k}.$
Note that $1/(k+1)^{-\hat\rho}\le C_j (\hat \rho)\le 1$ for every $j$, and hence, it follows that
\begin{equation}\label{Bound}
	\gamma+\EX\Big[\dfrac{\hat b_{n,k}(\hat \rho)}{(k+1)^{-\hat\rho}}\Big]\le \mu_{j} \le \gamma+\EX\Big[\hat b_{n,k}(\hat\rho)\Big] \hspace{6pt} \text{for every} \hspace{6pt} j\in\{1, 2, ..., k\}.
\end{equation}

This paper seeks to propose a reduced-bias  estimator that yields stable estimates over the number of top order statistics, $k.$ It is worth mentioning that \citet{Buitendag2018} has proposed the ridge regression as another reduced-bias estimator, which yields much more stable estimates compared to other existing estimators. However, our proposal is an alternative estimator that uses weighted least 
squares on the $Z_{j}'s$ in (\ref{E11}) to estimate $\gamma$ and $b_{n,k}$ jointly, and $\rho<0$ externally using the minimum variance method introduced by \cite{Buitendag2018}.

The concept of estimating the extreme value index using a weighted least squares estimator is not new in the field of extreme value theory.
Researchers  like \citet{Beirlant1996}, \citet{viharos1999weighted}, \citet{Huisman2001} and \citet{husler2006weighted} have used the weighted least squares method to estimate the extreme value index under different conditions. First,  \citet{Beirlant1996} fitted a weighted least squares estimator to the tail of the Pareto quantile plot and estimated $\gamma$ as the slope of the linear model. Second,  \citet{viharos1999weighted} linearised (3) and then fitted weighted least squares to their error sum of squares. Third,  \citet{Huisman2001} also employed the weighted least squares to estimate the extreme value index for heavy-tailed distributions whose slowly varying part is given as $\ell(x)=a\left(1+bx^{-\beta}\right)$, where $\beta>0$ and $a,b \in \mathbb{R}$. Lastly, \citet{husler2006weighted} estimated $\gamma$ by applying the weighted least squares method to 
$$\log\left(\dfrac{X_{n-i,n}}{X_{n.k,n}}\right)=\gamma\log\dfrac{k+1}{i+1},
$$
where $X_{n-i,n}=F_{n}^{\leftarrow}\left(1-\dfrac{i}{n}\right),i=0,1,...,k$.

Our proposed methodology differs from the existing studies in terms of; 
\begin{itemize}
	\item [i.] the definition of the weight function, and
	\item [ii.] the form of the regression model.
\end{itemize}
Hence, to the best of our knowledge, our proposed methodology has not been considered in literature.

Additionally, some authors have proposed other reduced-bias kernel estimators which utilises the weighted log-spacing for the estimation of $\gamma$ (see for example \cite{Beirlant2005},\cite{Caeiro2019a} and \cite{Caeiro2019b} )

The rest of the paper is organised as follows. The reduced-bias weighted least square estimator and its asymptotic properties with their proofs are presented in section~\ref{S2}. In section~\ref{S3}, the proposed estimator's performance is compared to other existing estimators via a simulation study and a practical illustration. Lastly, in Section 4, we provide concluding statements of the  study.\\

\section{\textbf{Materials and Methods}}\label{S2}

\subsection{The proposed estimator}
To jointly  estimate  $\gamma$ and $b_{n,k}$ in (\ref{E11}) via our proposed weighted least squares method, we define the loss function for the weighted least squares estimator as
\begin{equation}\label{E12}
L_{k}\left(\gamma,b_{n,k};W\right)=\sum_{j=1}^{k}W_{j}\left(Z_{j}-\gamma-b_{n,k}C_{j}\right)^2,
\end{equation}
for $j\in \{1,2,...,k\}$. The weight function is defined as 
\begin{equation}\label{E13}
W_{j}=1-\dfrac{j}{k+1}, 1\le j \le k,
\end{equation}
where $W_{j}\in (0,1)$ and decreases linearly with respect to $j$. We minimise (\ref{E12}) with respect to $\gamma$ and $b_{n,k}$ to obtain the following Weighted Least Squares (WLS) estimators:
\begin{equation}\label{BWLS}
\hat{b}_{n,k}=\dfrac{\sum_{j=1}^{k}\tilde{W}_{j}\left(C_{j}-\sum_{j=1}^{k}\tilde{W}_{j}C_{j}\right)Z_{j}}{\sum_{j=1}^{k}\tilde{W}_{j}C_{j}^{2}-\left(\sum_{j=1}^{k}\tilde{W}_{j}C_{j}\right)^2} 
\end{equation}
\begin{center}
	and
\end{center}
\begin{equation}\label{RWLSP}
\hat{\gamma}_{wls}^{+}=\sum_{j=1}^{k}\tilde{W}_{j}Z_{j}-\hat{b}_{n,k}\sum_{j=1}^{k}\tilde{W}_{j}C_{j},
\end{equation}
where
\begin{equation}\label{WT}
\tilde{W}_{j}=\dfrac{W_{j}}{\sum_{j=1}^{k}W_{j}}, 1\le j\le k 
\end{equation}
and  $C_{j}=(j/(k+1))^{-\hat{\rho}}$. Here,
$\tilde{W}_{j}$ is normalised and sum up to 1, (i.e., $0\le \tilde{W}_{j} \le 1$ and $\sum_{j=1}^{k}\tilde{W}_{j}=1$) and $\rho<0$ is a second-order parameter which is estimated externally using methods proposed \cite{Buitendag2018} and \cite{alves2003estimation}.

The parameter $\rho$ plays an essential role when dealing with optimisation in extreme value analysis. The speed of convergence of the limiting distribution of the extremes and also the asymptotic normality of estimators of $\gamma$ is regulated by (\ref{E10}). It can be observed that the rate of convergence of (\ref{E10}) is determined by $\rho$. The estimation of $\rho$ is also vital in controlling the bias component of most EVI estimators (see \citealp{Beirlant1999, feuerverger1999, Gomes2002, Gomes2008}, among others). Therefore, it is important to have a good estimator for $\rho$ for an efficient adaptive choice of the optimal $k$ for EVI estimators. 
 In the simulation study and the practical illustration, we considered the minimum variance approach introduced by \cite{Buitendag2018} and the \citet{alves2003estimation} estimators. However, we present the results for the minimum variance method since that yields much more stable EVI estimates for the practical datasets.

\subsubsection{Properties of the Proposed estimator}
In this section we demonstrate that, the proposed weighted least squares estimator,$\hat{\gamma}_{wls}^{+}$ is asymptotically unbiased, consistent and normal. 

First, we compute the asymptotic mean square error (AMSE) of the estimator, $\hat{\gamma}_{wls}^{+}$. It is well known that the mean square error (MSE) can be decomposed into a variance term and a bias term. The bias is defined as $\EX(\hat{\gamma}_{wls}^{+})-\gamma$, the distance between the estimator's expected value and the parameter $\gamma$. An estimator is said to be unbiased if the bias is $0$, in which case the MSE is just the variance of the estimator. The expected value of the  estimator $\hat{\gamma}_{wls}^{+}$ is given by
\begin{align*}
\EX\left(\hat{\gamma}_{wls}^{+}\right)
&=\sum_{j=1}^{k}\tilde{W}_{j}\EX(Z_{j})-\hat{b}_{n,k}\sum_{j=1}^{k}\tilde{W}_{j}C_{j}\\
&=\sum_{j=1}^{k}\tilde{W}_{j}\left(\gamma+\hat{b}_{n,k}C_{j}\right)-\hat{b}_{n,k}\sum_{j=1}^{k}\tilde{W}_{j}C_{j}\\
&=\gamma.
\end{align*}
Therefore, the asymptotic bias of $\hat{\gamma}^{+}_{wls}$, $Abias\left(\hat{\gamma}^{+}_{wls}\right)$ is
\begin{align*}
Abias(\hat{\gamma}_{wls}^{+})
&=\EX(\hat{\gamma}_{wls}^{+})-\gamma\\
&=0.
\end{align*}

Since $\hat{\gamma}_{wls}^{+}$ is asymptotically unbiased, the AMSE is the same as the asymptotic variance of the estimator, which we proceed to find.  Observe  from (\ref{BWLS}) and (\ref{RWLSP}) that  we can also rewrite $\hat{\gamma}_{wls}^+$ as
\begin{equation*}
\hat{\gamma}_{wls}^{+}=\sum_{j=1}^{k}\tilde{W}_{j}\left(1+\dfrac{S_{1}^2-S_{1}C_{j}}{S_{2}}\right)Z_{j},
\end{equation*} 
where
$S_{1}=\sum_{j=1}^{k}\tilde{W}_{j}C_{j}$  and 
$S_{2}=\sum_{j=1}^{k}\tilde{W}_{j}C_{j}^2-\left(\sum_{j=1}^{k}\tilde{W}_{j}C_{j}\right)^2$. Now, let  $\dot{S}:=\sum_{j=1}^{k}\tilde{W}_{j}^2\left(S_{1}-C_{j}\right)$ and
$\ddot{S}:=\sum_{j=1}^{k}\tilde{W}_{j}^2\left(S_{1}-C_{j}\right)^2$  and  observe  that  we  have

\begin{align*}
Var(\hat{\gamma}_{wls}^{+})
&=\sum_{j=1}^{k}\tilde{W}_{j}^2\left(1+\dfrac{S_{1}^2-S_{1}C_{j}}{S_{2}}\right)^2Var(Z_{j})\\
&=\gamma^2\left(\sum_{j=1}^{k}\tilde{W}_{j}^2+\dfrac{2S_{1}}{S_{2}}\sum_{j=1}^{k}\tilde{W}_{j}^2(S_{1}-C_{j})+\dfrac{S_{1}^2}{S_{2}^2}\sum_{j=1}^{k}\tilde{W}_{j}^2(S_{1}-C_{j})^2\right)\\
&=\gamma^2\left(\dfrac{4}{3k}+\dfrac{2S_{1}\dot{S}}{S_{2}}+\dfrac{S_{1}^2\ddot{S}}{S_{2}^2}\right)+O({1}/{k}),
\end{align*}
where  we  have  used $\displaystyle \lim_{k\to\infty} \Big[k\sum_{j=1}^{k}\tilde{W}_{j}^2\Big] =\lim_{k\to\infty} \Big[\frac{4}{k}\sum_{j=1}^{k} W_{j}^2\Big]=4\int_{0}^{1}(1-u)^2 du= 4/3 $  in  the  last  step.
Hence, the AMSE of the weighted least squares estimator is given by
\begin{equation}
0\le AMSE(\hat{\gamma}_{wls}^{+})=\gamma^2\left(\dfrac{4}{3k}+\dfrac{4S_{1}\dot{S}}{S_{2}}+\dfrac{S_{1}^2\ddot{S}}{S_{2}^{2}}\right)+O(1/k).
\end{equation}

Note  that, the following basic properties are required to study the asymptotic behaviour of the reduced-bias weighted least squares estimator.  
\begin{lemma}\label{L1}
	Assume that $\rho$ is estimated by a consistent estimator $\hat{\rho}$.  Then, as $k\rightarrow \infty$;
	\begin{itemize}
		\item[i.]
		$S_{1}=\sum_{j=1}^{k}\tilde{W}_{j}C_{j} \rightarrow 2/(1-\rho)(2-\rho)$.
		\item [ii.]
		$S_{2}=\sum_{j=1}^{k}\tilde{W}_{j}C_{j}^2-\left(\sum_{j=1}^{k}\tilde{W}_{j}C_{j}\right)^2 \rightarrow \rho^2(5-\rho)/(1-2\rho)(1-\rho)^2(2-\rho)^2$.
		\item[iii.]
		$\dot{S}=\sum_{j=1}^{k}\tilde{W}_{j}^2\left(S_{1}-C_{j}\right) \rightarrow 0$.
		\item[iv.]
		$\ddot{S}=\sum_{j=1}^{k}\tilde{W}_{j}^2\left(S_{1}-C_{j}\right)^2 \rightarrow 0$.
	\end{itemize}
\end{lemma}

\begin{proof}[Proof of Lemma 1]
	The following approximations are useful in the proof of Lemma~\ref{L1}:
	We  recall  the weight function  and  the  normalised  weight  function  as follows:
	\begin{equation*}
	W_{j}=1-\dfrac{j}{k+1}, 1\le j \le k\\
	\end{equation*}
	and
	$$\tilde{W}_{j}=\dfrac{W_{j}}{\sum_{j=1}^{k}W_{j}}$$
	
	Now the proof of Lemma~\ref{L1} is as follows:
	
	\begin{itemize}
		\item [i.]
		$S_{1}=\sum_{j=1}^{k}\tilde{W}_{j}C_{j}$ can be rewritten as $$S_{1}=\dfrac{2}{k}\sum_{j=1}^{k}W_{j}C_{j}=\dfrac{2}{k}\sum_{j=1}^{k}\left(1-\dfrac{j}{k+1}\right)\left(\dfrac{j}{k+1}\right)^{-\rho}.$$\\
		
		As \hspace{6pt} $k \rightarrow \infty$,  we  have 	
		\begin{align*}
		\dfrac{1}{k}\sum_{j=1}^{k}W_{j}C_{j}
		=\dfrac{1}{k}\sum_{j=1}^{k}\left(1-\dfrac{j}{k+1}\right)\left(\dfrac{j}{k+1}\right)^{-\rho}&= \int_{\dfrac{1}{k}}^{1}\left(1-u\right)u^{-\rho}\,d{u}+o(1)\\
		&=\dfrac{1}{(1-\rho)(2-\rho)}+O(1).
		\end{align*}
		Hence $S_{1} \rightarrow 2/(1-\rho)(2-\rho)$ as $k \rightarrow \infty$.
		\item [ii.]
		$S_{2}=\sum_{j=1}^{k}\tilde{W}_{j}C_{j}^2-\left(\sum_{j=1}^{k}\tilde{W}_{j}C_{j}\right)^2$ can also be expressed as $$S_{2}=\dfrac{2}{k}\sum_{j=1}^{k}W_{j}C_{j}^2-S_{1}^2=\dfrac{2}{k}\sum_{j=1}^{k}\left(1-\dfrac{j}{k+1}\right)\left(\dfrac{j}{k+1}\right)^{-2\rho}-\left(\dfrac{2}{(1-\rho)(2-\rho)}\right)^2.$$\\
		
		Also as  $k \rightarrow \infty$,
		\begin{align*}
		\dfrac{1}{k}\sum_{j=1}^{k}W_{j}C_{j}^2
		&=\dfrac{1}{k}\sum_{j=1}^{k}\left(1-\dfrac{j}{k+1}\right)\left(\dfrac{j}{k}\right)^{-\rho}
	=\dfrac{1}{2(1-\rho)(1-2\rho)}+o(1).
		\end{align*}

		This implies that, as $k \rightarrow \infty$;
		\begin{align*}
		S_{2}
		&=\dfrac{2}{2(1-\rho)(1-2\rho)}-\dfrac{4}{(1-\rho)^2(2-\rho)^2}+ o(1)=\dfrac{\rho^2(5-\rho)}{(1-2\rho)(1-\rho)^2(2-\rho)^2}+o(1).
		\end{align*}
		That is, as $k \rightarrow \infty$, $S_{2} \rightarrow \rho^2(5-\rho)/(1-2\rho)(1-\rho)^2(2-\rho)^2$.
		
		\item[iii.] The expression $\dot{S}=\sum_{j=1}^{k}\tilde{W}_{j}^2\left(S_{1}-C_{j}\right)$ can be rewritten as $$\dot{S}=\dfrac{4}{k^2}\sum_{j=1}^{k}W_{j}^2\left(S_{1}-C_{j}\right).$$ 
		Note  from  $\displaystyle 	0\le W_{j}\le 1 \text{and} \hspace{6pt} 0\le C_{j}\le 1$ that we have
		
		$$-\sum_{j=1}^{k}W_{j}^2C_{j}\le \sum_{j=1}^{k}W_{j}^2(S_{1}-C_{j})\le 2\sum_{j=1}^{k}W_{j}^2$$
		$$ -\dfrac{4}{k^2}\sum_{j=1}^{k}W_{j}^2C_{j}\le \dfrac{4}{k^2}\sum_{j=1}^{k}W_{j}^2(S_{1}-C_{j})\le \dfrac{8}{3k}.$$
		
		Therefore,  we  have,  
		$\displaystyle  -\dfrac{4}{k}\le \dfrac{4}{k^2}\sum_{j=1}^{k}W_{j}^2(S_{1}-C_{j})\le \dfrac{8}{3k},$
	
		which  gives	$0\le\lim\limits_{k\rightarrow \infty}\dfrac{4}{k^2}\sum_{j=1}^{k}W_{j}^2(S_{1}-C_{j})\le 0.$
		Hence $\dot{S} \rightarrow 0$ as $k\rightarrow \infty$.

		\item[iv.] We can also express $\ddot{S}=\sum_{j=1}^{k}\tilde{W}_{j}^2\left(S_{1}-C_{j}\right)^2$ as $$\ddot{S}=\dfrac{4}{k^2}\sum_{j=1}^{k}W_{j}^2\left(S_{1}-C_{j}\right)^2.$$
		
		We  observe  that  
		$\displaystyle  C_{j}^2\le (S_{1}-C_{j})^2\le 4$
		and  hence, we  have   $$\dfrac{4}{k^2}\sum_{j=1}^{k}W_{j}^2C_{j}^2\le\dfrac{4}{k^2}\sum_{j=1}^{k}W_{j}^2(S_{1}-C_{j})^2\le \dfrac{16}{k^2}\sum_{j=1}^{k}W_{j}^2\le\dfrac{16}{3k} .$$
		
		Now,  using 
		\begin{equation}\label{KDA1}
		0\le \dfrac{4}{k^2}\sum_{j=1}^{k}W_{j}^2C_{j}^2\le\dfrac{4}{k}
		\end{equation}
		
		in  \eqref{KDA1},  we  obtain 
		$0\le \dfrac{4}{k^2}\sum_{j=1}^{k}W_{j}^2(S_{1}-C_{j})^2\le \dfrac{16}{3k}$    and 
		$ 0\le \lim\limits_{k\rightarrow \infty}\dfrac{4}{k^2}\sum_{j=1}^{k}W_{j}^2(S_{1}-C_{j})^2\le 0 .$ 
		Thus  we  have , $\ddot{S} \rightarrow 0$ as $k\rightarrow \infty$ which  completes the proof of Lemma~\ref{L1}.
	\end{itemize}
\end{proof}

A desirable property of an estimator is consistency. If more observations are included, we hope to obtain a lot of information about the unknown parameter. From Lemma~\ref{L1}, the AMSE of $\hat{\gamma}_{wls}^{+}$ approaches 0, as $k \to \infty$. This implies that the estimator $\hat{\gamma}_{wls}^{+}$ is ``MSE-Consistent".\\

We conclude this section  by  stating and proving  the sampling  distribution  of the weighted least squares estimator, $\hat{\gamma}_{wls}^{+}$.   We write    $$S(\rho):=(1-2\rho)(1-\rho)^2(2-\rho)^2/\rho^2(5-\rho).$$

\begin{theorem}\label{TH1}
	Suppose (\ref{E4}) and (\ref{E10}) holds. Assume also that $\rho$ is estimated by a consistent estimator $\hat \rho$ with $\EX[S(\hat\rho)]<\infty.$  Let  $k \to \infty, k/n\to 0,$   as $n\to \infty$  and $\sqrt{k}b_{n,k}\to_{p} 0.$ Then  
	\begin{equation}\label{E18}
	\sqrt{3k}(\hat{\gamma}_{wls}^{+}-\gamma)/2\gamma \to_{d}N(0,1).
	\end{equation}
\end{theorem}

The following five  Lemmas will be required to prove Theorem (\ref{TH1}).  We  write  \begin{equation}\label{E19}
\mathscr{M}_{k}(C_{j}):=\dfrac{\tilde{W}_{j}(C_{j}-\sum_{j=1}^{k}\tilde{W}_{j}C_{j})}{\sum_{j=1}^{k}\tilde{W}_{j}C_{j}^2-\left(\sum_{j=1}^{k}\tilde{W}_{j}C_{j}\right)^2}.
\end{equation}

\begin{lemma}\label{L2}
	Let  $C_{j}=\left(\dfrac{j}{k+1}\right)^{-\rho}, \tilde{W}_{j}=\dfrac{W_{j}}{\sum_{j=1}^{k}W_{j}},j\in\{1,2,...,k\}$ and $\rho<0.$   Then   as $k\to \infty$,

	$$\mathscr{M}(C_{j})\to O(1/k^{1+\omega}),$$
	
	where   $0<\omega< 0.1 $.
\end{lemma}

\begin{proof}[Proof of Lemma 3]  We  observe that  $$-k^{\rho}/|k^{1+\omega+\rho}+n^\rho| \le \mathscr{M}_{k}(C_{j})\le k^{\rho}C_{j}^2/S_2|k^{1+\omega+\rho}+n^\rho|.$$
	Therefore, we  have  $\mathscr{M}_{k}(C_{j}) \to O(1/k^{1+\omega}),0<\omega<0.1$ as $k\to \infty$,  which  completes  the  proof  of  Lemma~\ref{L2}.
\end{proof}

\begin{lemma}\label{L4}
	  Suppose that $Z_{1}, Z_{2},Z_3,...$ are  such  that    $\hat\rho<0$  is  consistent  estimator of  $\rho<0$ and        $\EX\Big[S(\hat\rho)\Big]<\infty$. Then 
	\begin{itemize}
		\item[(i)]  conditional  on  $\{\hat\rho=\rho\}$  the moment generating function (m.g.f) of  $\hat b_{n,k}(\hat\rho)$is given as
		$$	M_{\hat{b}_{n,k}(\rho)}(kt)=\prod_{j=1}^{k}M_{Z_{j}}\left(k\mathscr{M}_{k}(C_{j})t\right)=\prod_{j=1}^{k}\dfrac{1}{1-\mu_{j}k\mathscr{M}_{k}(C_{j})t}.$$
	.

\item[(ii)]	 for any $\epsilon>0$ 
	\begin{equation}
	\lim\limits_{k \rightarrow \infty}\prob\left(\sqrt{k}\hat{b}_{n,k}(\hat{\rho})\ge \epsilon\right)=0.
	\end{equation}

\end{itemize}
\end{lemma}

\begin{proof}[\bf Proof of Lemma~\ref{L4}](i)
	 Observe  that  conditional  on   $\{\hat\rho=\rho\}$, the  sequence  of  random  variables   $Z_j=\gamma+b_{k,n}C_j+\epsilon_j$ are  independent  exponential  distributed  with mean $\mu_j(\rho)=\EX(Z_j|\hat\rho=\rho)=\gamma+\hat{b}_{n,k}C_j.$ Note  that  by  $\hat \rho$  is  a  consistent  estimator  of  $\rho$  with   $\EX\Big[S(\hat\rho)\Big]<\infty$ and  the  Minkowski's  inequality  we  have   $\EX(|Z_i|^p) <\infty$   for all $ p>0$.   Let   $	Y=a_{1}Z_{1}+a_{2}Z_{2}+...+a_{n}Z_{n}
	  =\sum_{j=1}^{n}a_{i}Z_{i},$ where  \, $a_{i} \in \mathbb{R}, i \ge 1.$  Then  conditional  m.g.f of   $Y$ given   $\{\hat\rho=\rho\}$  is 
	
	\begin{align}
	M_{Y,\rho}(t)  \nonumber
	=\EX (e^{Yt}\big|\hat\rho=\rho)&=\EX\left(e^{t\left(a_{1}Z_{1}+a_{2}Z_{2}+...+a_{n}Z_{n}\right)}\Big|\hat \rho=\rho\right)\\ \nonumber
	&=\EX\left(e^{a_{1}Z_{1}t}e^{a_{2}Z_{2}t}... e^{a_{n}Z_{n}t}\Big|\hat\rho=\rho\right)\\ \nonumber 
	&=\prod_{i=1}^{n}M_{Z_{i},\rho}(a_{i}t) \label{MGF}. 
	\end{align}
	
Now,   observe  from (\ref{E19})  and  (\ref{BWLS})   that $\hat{b}_{n,k}(\hat\rho)$ is   expressible  as  
	$\displaystyle 
	\hat{b}_{n,k}(\hat\rho)=\sum_{j=1}^{k}\mathscr{M}_{k}(C_{j})Z_{j}.
	$  	
Therefore,  the  moment generating function of $Z_{k}$  given  $\{\hat \rho=\rho\}$   is  
	$
	M_{Z_{k}}(t)=1/(1-\mu_{k}t)
	$
	and  the moment generating function of $\hat{b}_{n,k}$ of  given  $\displaystyle \{\hat \rho=\rho\}$   is obtained  by
	\begin{align*}
	M_{\hat{b}_{n,k}(\rho)}(t)=\prod_{j=1}^{k}M_{Z_{j}}\left(\mathscr{M}_{k}(C_{j})t\right)	\end{align*}
	
	\begin{equation*}
	\therefore M_{\hat{b}_{n,k}(\rho)}(kt)=\prod_{j=1}^{k}\dfrac{1}{1-\mu_{j}(\rho)k\mathscr{M}_{k}(C_{j})t}.
	\end{equation*}

\end{proof}

\begin{proof}[\bf Proof of Lemma~\ref{L4}](ii).Note  from Lemma~\ref{L2}  that, $\mathscr{M}_{k}(C_{j})\to O(1/k^{1+\omega})$ as $k\to \infty$; hence we   have  $\mu_{j}k\mathscr{M}_{k}(C_{j})\to O(1/k^{\omega})$ as $k \to \infty$,  $k/n \to 0$, as $n\to \infty.$  Therefore, we  can  find  some $ \delta>0$ such that
	$\displaystyle -\delta \le\mu_{j}k\mathscr{M}_{k}(C_{j})\le \delta.$  
	This  implies, 

	$$ \log\left(1-\delta+\delta^2-\delta^3+...\right)\le \lim\limits_{k \rightarrow \infty} \dfrac{1}{k}\log M_{\hat{b}_{n,k}}(kt) \le \log\left(1+\delta+\delta^2+\delta^3+...\right)$$

	Now  taking  the  limit  $\delta\downarrow 0$  and  applying  the  Sandwich Theorem,  we  have    
	
	\begin{equation}\label{E23}
	\lim\limits_{k \rightarrow \infty} \dfrac{1}{k}\log M_{\hat{b}_{n,k}(\rho)}(kt)=0.
	\end{equation}
\end{proof}

Equation~(\ref{E23}) means that the limiting logarithmic conditional  moment generating function converges to $0$ with speed $k$. Hence by the G$\ddot{a}$rther-Ellis theorem \citep[Theorem 2.3.6, p.~44]{Dembo1998},  $\displaystyle \{\hat \rho=\rho\},$ conditional  on the  $\hat{b}_{n,k}$   satisfies a large deviation principle (LDP) in the space of non-negative real numbers with speed $k$ and rate function $I(x)$ given  by  

\begin{align*}
I(x)
&=\sup\limits_{\lambda\in \mathbb{R}}\left\{\lambda x-0\right\}\\
&=
\begin{cases}
0 \hspace{8pt} \text{if} \hspace{6pt} x=0
\\
\infty \hspace{6pt} \text{if} \hspace{6pt} x\neq 0
\end{cases}
\end{align*}
where $x\in [0,\infty)$ and $\lambda>0$. If   conditional  on the event $\{\hat\rho=\rho\}$,  $\hat b_{n,k}(\hat\rho)$ obeys LDP  in  the  space  $[0,\infty)$ with speed $k$ and a  rate function $I(x)$, then for every  $\epsilon >0$,  we  have 
\begin{equation}
\begin{aligned}
&\limsup_{k\to\infty}\frac{1}{k}\log\mathbb{P}\left(\sqrt{k}\hat{b}_{n,k}(\hat\rho)>\epsilon\right)\\
&=\limsup_{k\to\infty}\frac{1}{k}\log\mathbb{P}\left(\hat{b}_{n,k}(\hat\rho)> \tau_k\Big| \, \hat\rho=\rho\right)+\limsup_{k\to\infty}\frac{1}{k}\log\mathbb{P}\Big(\hat\rho=\rho\Big)\\
&\le -\inf_{x\in (0,\infty)}I(x)+\limsup_{k\to\infty}\frac 1 k\log (1)\\
&=-\inf_{x\in (0,\infty)}I(x),
\end{aligned}
\end{equation}

where $\tau_k= \dfrac{\epsilon}{\sqrt{k}}.$
As  the  non- typical behaviour of the rate function is when $x\ne0$, we  have  
\begin{equation*}
\lim\limits_{k \to \infty}\mathbb{P}\left(\sqrt{k}\hat{b}_{n,k}(\hat\rho)>\epsilon\right)\le 0.
\end{equation*}
Thus, $\sqrt{k}\hat{b}_{n,k}(\hat\rho)\to_{p} 0,$ as $k\to \infty$  which  ends  the  Proof  of  Lemma~\ref{L4}. 

\begin{lemma}\label{L5}
Suppose that $Z_{1}, Z_{2},...$ are  such  that    $\hat\rho<0$  is  consistent  estimator of  $\rho<0$ and    $\EX\Big[S(\hat\rho)\Big]$   is  finite.   
Then,
\begin{itemize}

	\item  [(i)]    $\EX (Z_{j})=\mu_{j}<\infty,$ \,  and     $\lim_{k\to\infty}Var(Z_{k})=\gamma^2,$

 \item [(ii)] and   there exists $\delta>0,$ such that  
$\displaystyle  \lim\limits_{k \rightarrow \infty}\dfrac{1}{S_{n}^{2+\delta}}\sum_{j=1}^{k}\EX\left(|Z_{j}-\mu_j|^{2+\delta}\right)=0.$
	
	\end{itemize}
\end{lemma}
\begin{proof}[Proof of Lemma~\ref{L5}]
We  observe  that  $Z_{1}, Z_{2}, Z_{3},...$ are independent but not identical  distributed  random  variables.
\begin{itemize}
	\item [(i)] $$\begin{aligned}  \EX(Z_j)=\EX\Big\{\EX\left(Z_{j}\big| \hat \rho\right)\Big\}&= \EX\Big\{\gamma+\hat b_{n,k}(\hat\rho)C_j(\hat\rho)\Big\}\le\gamma+\EX(\hat b_{n,k}(\hat\rho))\\
	& \le\gamma+ \EX\Big[\frac{(1-2\hat\rho){(1-\hat\rho)^2}(2-\hat\rho)^2\sum_{j=1}^{k}C_j(\hat\rho)Z_j/k}{ (\hat\rho^2(5-\hat\rho))|k^{1+\omega+\hat\rho}+n^{\hat\rho}|}+o(1)\Big]\\
	&\le\gamma+\EX\Big[\frac{(1-2\hat\rho){(1-\hat\rho)^2}(2-\hat\rho)^2}{\hat\rho^2(5-\hat\rho))k^{\omega}} \EX\Big[\sum_{j=1}^{k}C_j(\hat\rho)Z_j/k \Big]+o(1)\,\,\, \mbox{a.s.}\\
		&\le\gamma+\EX\Big[\frac{(1-2\hat\rho){(1-\hat\rho)^2}(2-\hat\rho)^2}{\hat\rho^2(5-\hat\rho)}\Big] \Big[\frac{\EX(\gamma_k^{H})}{k^{\omega}}\Big] +o(1)\,\,\, \mbox{a.s.}\\
	&\le\gamma+ o(1/k^{\omega})+o(1), 	
	\end{aligned}$$
	
	where  $0<\omega<0.1$    and  $\gamma_k^{H}$  is  the  Hill  estimator.
	\item[(ii)] Let  $f$  be  the  probability  density  function  of  $ Z_j  $   and  observe  that  we
	
	\begin{align*}
	\EX\left(|Z_{j}-\mu_{j}|^{2+\delta}\big| \hat \rho=\rho\right)
	&=\int_{0}^{\infty}|Z_{j}-\mu_{k}|^{2+\delta}f(z_{j})\,d{z}_{j}\\
	&=-\int_{0}^{\mu_{j}}\left(Z_{j}-\mu_{j}\right)^{2+\delta}f(z_{j})\,d{z}_{j}+\int_{\mu_{j}}^{\infty}\left(Z_{j}-\mu_{kj}\right)^{2+\delta}f(z_{j})\,d{z}_{j}\\
	&=\dfrac{e^{-1}}{\mu_{j}}\left\{\dfrac{(-\mu_{j})^{3+\delta}(7+2\delta)}{(3+\delta)(4+\delta)}+\mu_{j}^{3+\delta}(2+\delta)!\right\}\{1+O(1)\}\\
	&=\mu_{j}^{2+\delta}\eta(\delta)\{1+O(1)\},
\end{align*}
where 
$\eta(\delta)=\dfrac{1}{e}\left\{\dfrac{(-1)^{3+\delta}(7+2\delta)}{(3+\delta)(4+\delta)}+(2+\delta)!\right\}<\infty$ for  $0<\delta<\infty$.  Therefore  we  have  that   $$\EX\left(|Z_{j}-\mu_{j}|^{2+\delta}\right)=\EX\Big\{\EX\left(|Z_{j}-\mu_{j}|^{2+\delta}\big| \hat \rho\right)\Big\}=\eta(\delta)\{1+O(1)\} \mu_{j}^{2+\delta}.  $$

	Now define
	$T_{j}=Z_{j}-\mu_{j}, 1\le j \le k,$   and  note  that  
	
	$$S_{n}^2 =Var\left(\sum_{j=1}^{k}T_{j}\right)=\sum_{j=1}^{k}Var(Z_{j}-\mu_{j})
	=\sum_{j=1}^{k}Var(Z_{j})=\sum_{j=1}^{k}Var(\epsilon_{j})= k\gamma^2.$$

	Hence,
	$$\begin{aligned}
	\lim\limits_{k\to \infty}\dfrac{1}{S_{n}^{2+\delta}}\sum_{j=1}^{k}\EX\left(|Z_{j}-\mu_{j}|^{2+\delta}\right)
	&\le\lim\limits_{k \to\infty}\dfrac{k\eta(\delta)\{1+o(1)\}\Big(\gamma+ o(1/k^{\omega})+o(1)\Big)^{2+\delta}}
	{(k\gamma^2)^{1+\delta/2}}\\
	&=0
	\end{aligned}$$
	
\end{itemize}	
\end{proof}

We write  $\displaystyle \lim_{k\to\infty}\EX\big(Y_k):=\mu$,   $\displaystyle \lim_{k\to\infty}Var(Y_k):=\sigma^2$  and  observe  from   Lemma (\ref{L4})  and (\ref{L5})  the  Proof of  Theorem~\ref{TH1}  as  follows: Lemma (\ref{L4}) establishes that $\sqrt{k}\hat{b}_{n,k}$ converges in probability to $0$ as $k$ becomes large. Lemma (\ref{L5}) shows that Lyapunov's condition for the central limit theorem holds \citep[p.~359-362]{Billingsley1983}, that is, $Y_k:=\sqrt{3k}\left(\hat{\gamma}_{wls}^{+}-\gamma\right)/{2\gamma} \to_{d}N(\mu,\sigma^2\big)$, as $k\to \infty$. Therefore, all we need to complete the proof of Theorem~\ref{TH1} is to estimate the parameters $\mu$ and $\sigma^2$ under the condition $\sqrt{k}\hat{b}_{n,k}\to_{p} 0 $.\\

Recall  that  $Y_{k}$ converges in distribution to normally distributed  random  variable with mean $\mu$ and variance $\sigma^2.$ 
We  compute the  parameters $\mu$  and  $\sigma^2$  as follows:
\begin{align*}
\mu=\lim_{k\to\infty}\EX\big(Y_k)&=\lim_{k\to\infty}\Big[\frac{\sqrt{3k}}{2\gamma}\Big(\EX(\hat{\gamma}_{wls}^{+})-\gamma\Big)\Big]\\
&=\lim_{k\to\infty}\Big[\frac{\sqrt{3k}}{2\gamma}\Big(\sum_{j=1}^{k}\tilde{W}_{j}\EX(Z_{j})-\gamma\Big)\Big] \hspace{18pt}    (\text{since} \hspace{6pt} \sqrt{k}\hat{b}_{n,k}\to 0 \hspace{6pt} \text{as} \hspace{6pt} k\to \infty)\\
&=\lim_{k\to\infty}\Big[\frac{\sqrt{3k}}{2\gamma}\sum_{j=1}^{k}\gamma\tilde{W}_{j}+\frac{\sqrt{3k}}{2\gamma}\hat{b}_{n,k}\sum_{j=1}^{k}\tilde{W}_{j}C_{j}-\frac{\sqrt{3k}}{2\gamma}\gamma\Big]=0\end{align*}

Further,  we  have 
$$\displaystyle \sigma^2=\lim_{k\to\infty}Var(Y_{k})
=\lim_{k\to\infty}\Big[\frac{3k}{4\gamma^2}\sum_{j=1}^{k}\tilde{W}_{j}^{2}Var(Z_{j})\Big]=\lim_{k\to\infty}\Big[\frac{3k}{4}\sum_{j=1}^{k}\tilde{W}_{j}^2\Big]=\lim_{k\to\infty}\Big[ 1+O(1/k)\Big]=1.$$ 
Hence $\sqrt{k}\hat{b}_{n,k}\to_{p} 0$  and $Y_{k}\to_{d} N(0,1),$ as $k\to \infty$, as  required.

\section{\textbf{Results and Discussion}}\label{S3}
We present a simulation study that compares the performance of the proposed extreme value index estimator with some existing estimators in the literature.

\begin{scriptsize}
\begin{table}[!htbp]
	\caption{Heavy-tailed distributions from the Pareto-type distribution}
	\centering
	\begin{tabular}{*{7}{c}}
		\hline
		Distribution & $1-F(x)$ & $\ell_{F}(x)$ & $\gamma$\\ \hline 
		Burr & $\left(1+x^{\tau}\right)^{-\lambda}$ & $\left(1+x^{-\tau}\right)^{-\lambda}$ & $1/\lambda\tau$ \\
		Fr\'{e}chet & $1-exp(-x^{-\alpha})$ & $1-\frac{x^{-\alpha}}{2}+O(x^{-\alpha})$ & $1/\alpha$ \\ 
		Log-Gamma & $\int_{x}^{\infty}\frac{\lambda^{\alpha}}{\varGamma(\alpha)}w^{\lambda-1}\left(\log w\right)^{\alpha-1}\,dw$ & $\frac{\lambda^{\alpha-1}}{\varGamma(\alpha)}\left(\log x\right)^{\alpha-1}$$\left(1+\frac{\alpha-1}{\lambda}\frac{1}{\log x}+0\left(\frac{1}{\log x}\right)\right)$ & $1/\lambda$ \\ \hline
	\end{tabular}
	\label{Table:T1}
\end{table}
\end{scriptsize}

\subsection{Simulation study }
In this simulation study, the reduced-bias weighted least square estimator is compared to the Hill \citep{hill1975}, the least square \citep{Beirlant2002}, the Bias-corrected Hill \citep{Caeiro2005a} and the ridge regression \citep{Buitendag2018} estimators. The estimators are illustrated for three distributions from the Pareto-type distribution as shown in Table~\ref{Table:T1}, taking 1000 repetitions of samples of sizes 50 and 200. The bias, MSE and the average EVIs are plotted as a function of $k$ to investigate the behaviour of the estimators' sample paths. For each distribution, we consider three different values of $\gamma \in (0,1]$ to investigate the performance of the proposed estimator for varying values of $\gamma$. In particular, we consider $\gamma = 0.10, 0.50 $ and $1.00$, since these values are usually used in simulation studies (see for example \citealp{minkah2021robust,Beirlant2019,minkah2018extreme, cabral2020comparison}). The following distributions are used:
\begin{enumerate}
	\item[$\bullet$] Burr$\left(\eta,\tau,\lambda\right)$
	\begin{itemize}
		\item[\checkmark]$\eta=1, \tau=\sqrt{10}$ and $\lambda=\sqrt{10}$, so that $\gamma=0.10$.
		\item[\checkmark]$\eta=1, \tau=\sqrt{2}$ and $\lambda=\sqrt{2}$, so that $\gamma=0.50$.
		\item[\checkmark]$\eta=1, \tau=2$ and $\lambda=1/2$, so that $\gamma=1.00$.
	\end{itemize}
\item[$\bullet$] Fr\'{e}chet$(\alpha)$
\begin{itemize}
	\item[\checkmark]$\alpha=10,2$ and $1$, so that $\gamma=0.10, 0.50$ and $1.00$, respectively.
\end{itemize}
\item[$\bullet$]Log Gamma$\left(\lambda,\alpha\right)$
\begin{itemize}
	\item[\checkmark]$\lambda=10$ and $\alpha=2$, so that $\gamma=0.10$
	\item[\checkmark]$\lambda=2$ and $\alpha=2$, so that $\gamma=0.50$
	\item[\checkmark]$\lambda=1$ and $\alpha=2$, so that $\gamma=1.00$
\end{itemize}
\end{enumerate}


Table~\ref{Table:T2} presents the notations of the extreme value index estimators used in the simulation study.

\begin{table}[!htbp]
	\caption{Notations of the Estimators}
	\begin{center}
		\begin{tabular}{lc}
			
			\hline
			\textbf{Estimators} & \textbf{Notation}\\
			\hline
			Hill & HILL\\
			Bias-corrected Hill & BCHILL \\
			Least square & LS \\
			Ridge regression & RR\\
			reduced-bias Weighted least square & WLS \\
			\hline
		\end{tabular}
		\label{Table:T2}
	\end{center}
\end{table}

The simulation results for the Burr distribution are presented in Figures~\ref{Fig:Figure 1} - \ref{Fig:Figure 3}, and that of the Fr\'{e}chet distribution are shown in Figures~\ref{Fig:Figure 4} - \ref{Fig:Figure 6}.The sample paths of the proposed estimator, WLS, are close to that of LS in most cases, and this means that the two estimators are competitively close to each other. 

In the case of the Burr distribution, WLS generally outperforms the other estimators in terms of MSE and bias for small EVI ($\gamma=0.1$), and the WLS estimator is mostly the second-best estimator in terms of bias and MSE for $\gamma=0.50$. However, in the case of $\gamma=1.00$, the sample paths of WLS and BCHILL are generally close for small to medium values of $k$. The WLS estimator mostly has the lowest bias and MSE among the three reduced-bias estimators (i.e., WLS, LS and RR) for the Burr distribution.

Furthermore, in the case of the Fr\'{e}chet distribution, the WLS is mostly the best performing estimator in terms of bias. The sample paths of WLS and LS are mainly close to each other. However, WLS generally appears to outperform LS in terms of bias. For a large sample, i.e., $n=200$, the WLS estimator outperforms the other estimators in terms of MSE for large values of $k$.

Additionally, the performance of the estimators on samples generated from the Log-Gamma distribution is presented in the appendix. The proposed estimator, WLS, generally is the second-best performing estimator to the BCHILL estimator in terms of bias and MSE across all samples.


The MSE plots of the WLS estimator are stable over the middle region of $k$ and thus makes the determination of the optimal $k$, defined by $k_{0}=\arg\min_{k}$$\left[MSE\{\hat{\gamma}_{wls}^{+}(k)\}\right]$, easier.

Given the role of $\rho$, the WLS was also compared to the other estimators, using the \citet{alves2003estimation} estimator for $\rho$, and although it was not universally the best in terms of bias and MSE, it performed well compared to the other estimator. However, for brevity we omit these results.

In summary, the reduced-bias weighted least squares estimator generally yields lower bias and MSE, which are stable over a long range of $k$ values. Thus, it can be considered an appropriate estimator of the extreme value index for samples generated from the Pareto-type distribution.

\begin{figure}[!htbp]
	\centering
	\subfloat{\includegraphics[width=5.5cm,height=4cm]{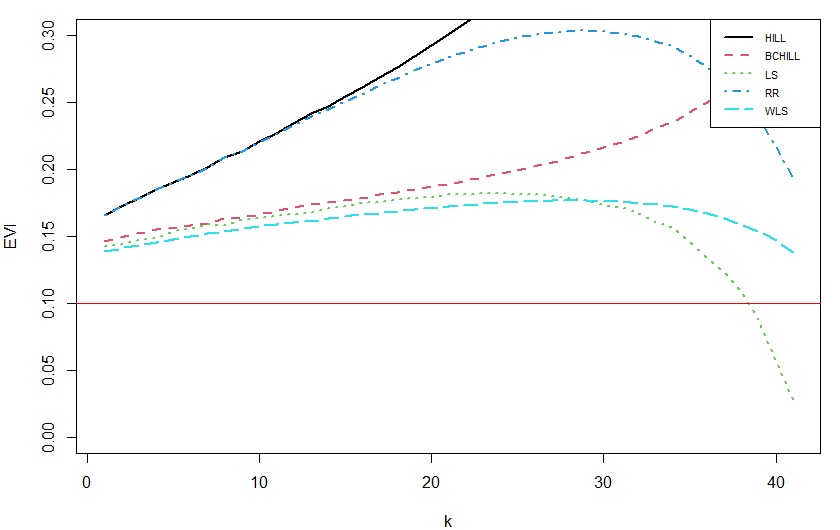}}\hfill
	\subfloat{\includegraphics[width=5.5cm,height=4cm]{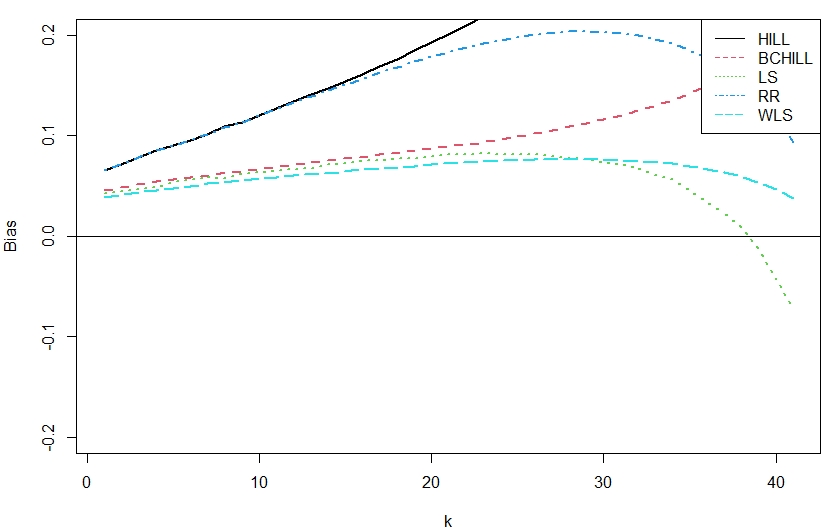}}\hfill
	\subfloat{\includegraphics[width=5.5cm,height=4cm]{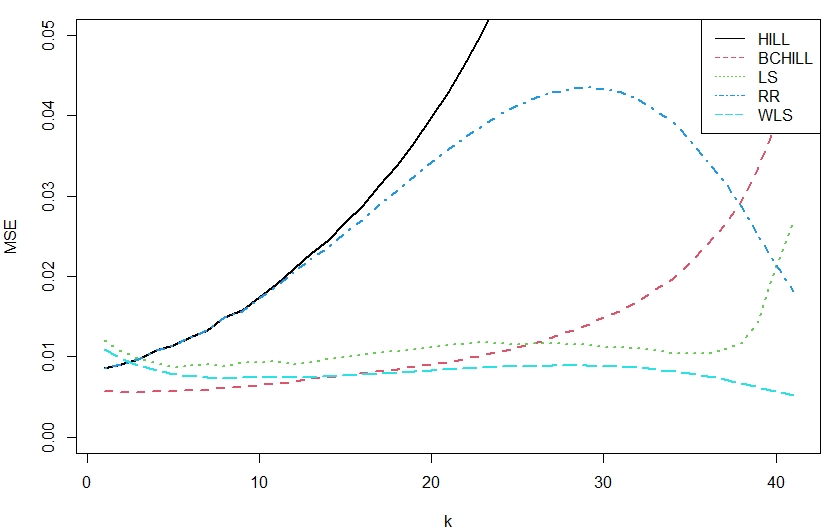}}\hfill
	\\[\smallskipamount]
	\subfloat{\includegraphics[width=5.5cm,height=4cm]{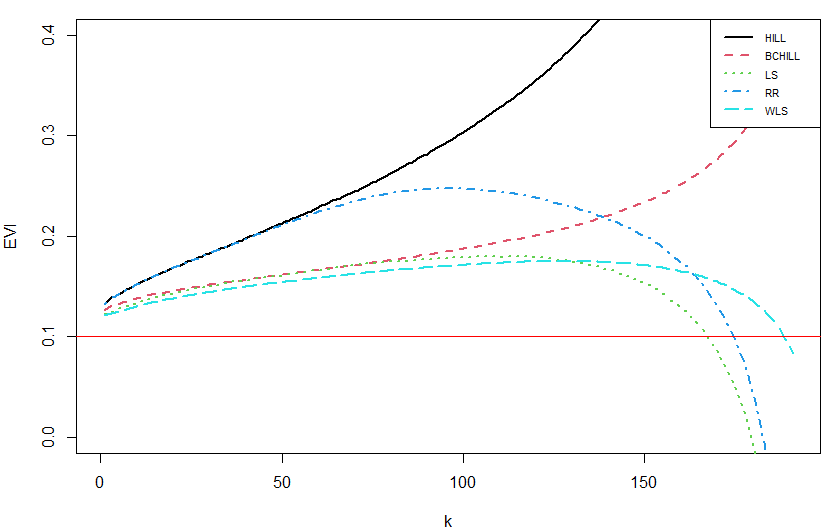}}\hfill
	\subfloat{\includegraphics[width=5.5cm,height=4cm]{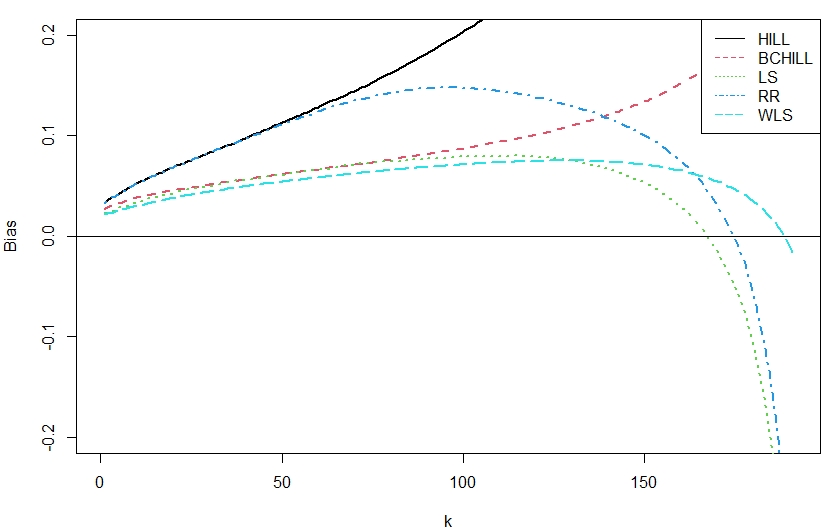}}\hfill
	\subfloat{\includegraphics[width=5.5cm,height=4cm]{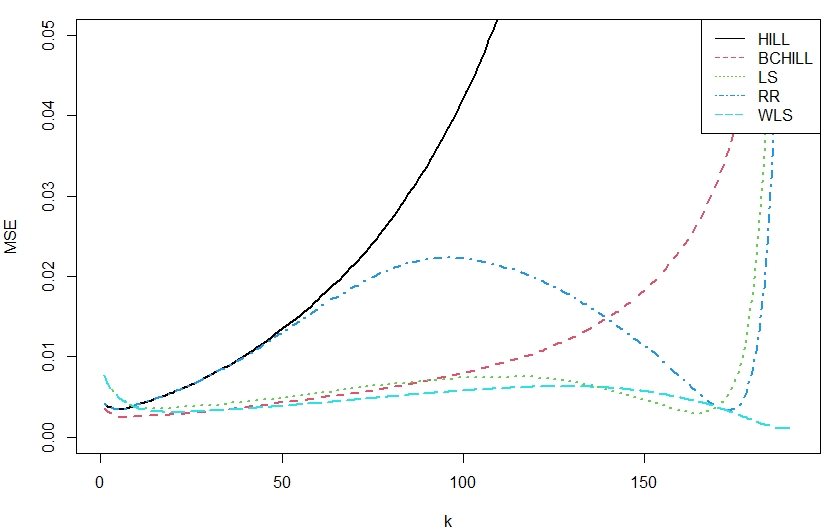}}\hfill
	\\[\smallskipamount]
	\caption{Results for Burr distribution with $\gamma=0.1$: average EVI(leftmost panel); Bias(middle panel); MSE(rightmost panel). First row: $n=50$; second row: $n=200$}
	\label{Fig:Figure 1}
\end{figure}
\pagebreak
\begin{figure}[!htbp]
	\centering
	\subfloat{\includegraphics[width=5.5cm,height=4cm]{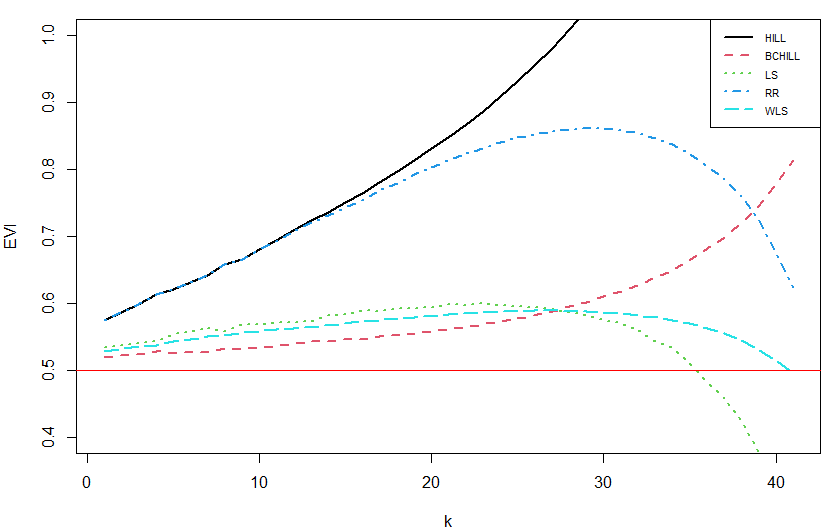}}\hfill
	\subfloat{\includegraphics[width=5.5cm,height=4cm]{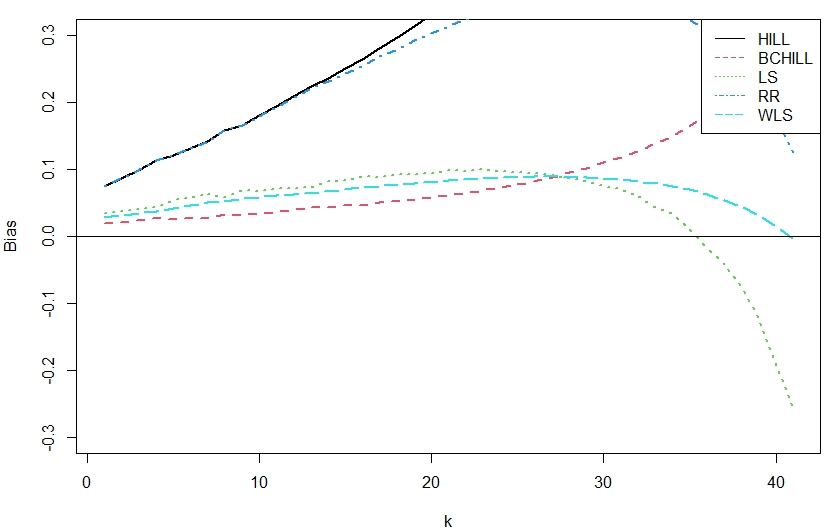}}\hfill
	\subfloat{\includegraphics[width=5.5cm,height=4cm]{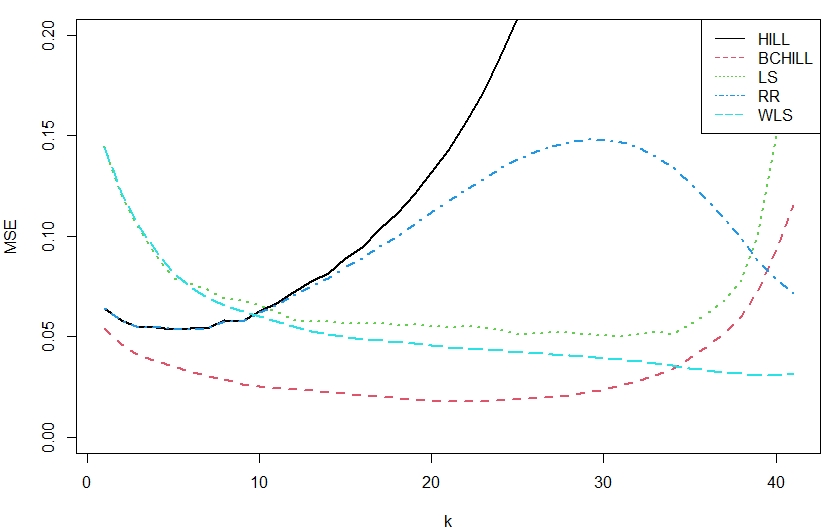}}\hfill
	\\[\smallskipamount]
	\subfloat{\includegraphics[width=5.5cm,height=4cm]{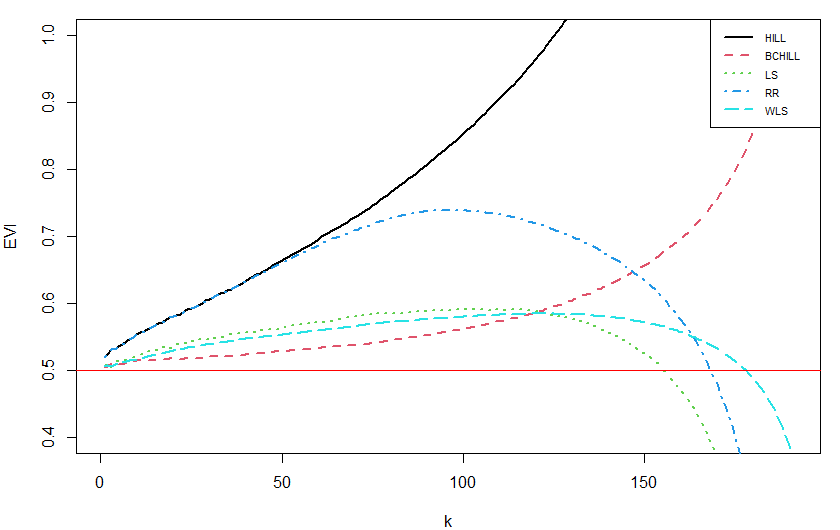}}\hfill
	\subfloat{\includegraphics[width=5.5cm,height=4cm]{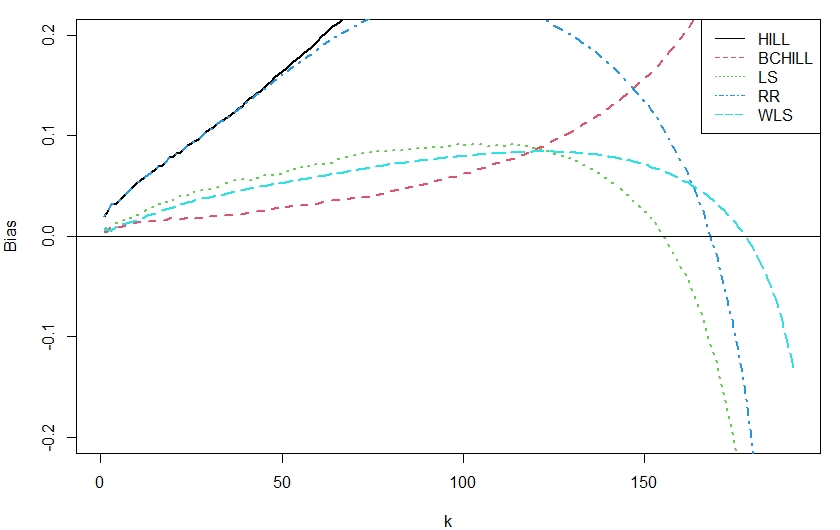}}\hfill
	\subfloat{\includegraphics[width=5.5cm,height=4cm]{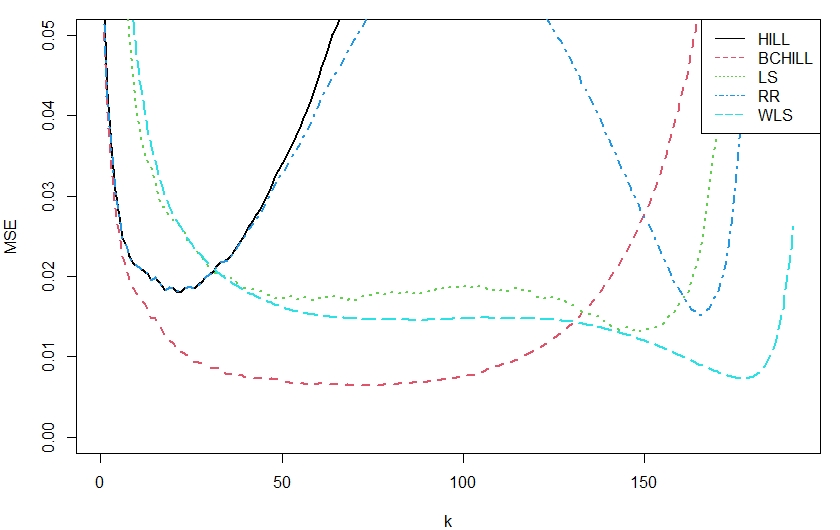}}\hfill
	\caption{Results for Burr distribution with $\gamma=0.5$: average EVI(leftmost panel); Bias(middle panel); MSE(rightmost panel). First row: $n=50$; second row: $n=200.$}
	\label{Fig:Figure 2}
\end{figure}
\begin{figure}[!htbp]
	\centering
	\subfloat{\includegraphics[width=5.5cm,height=4cm]{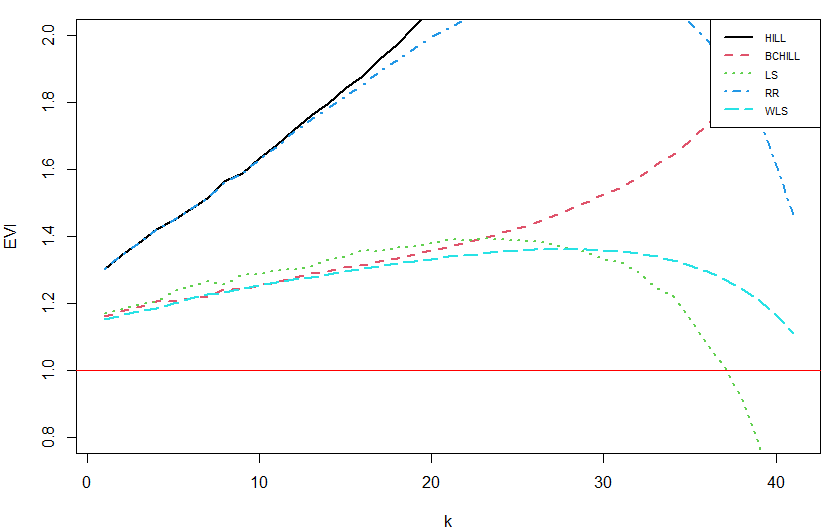}}\hfill
	\subfloat{\includegraphics[width=5.5cm,height=4cm]{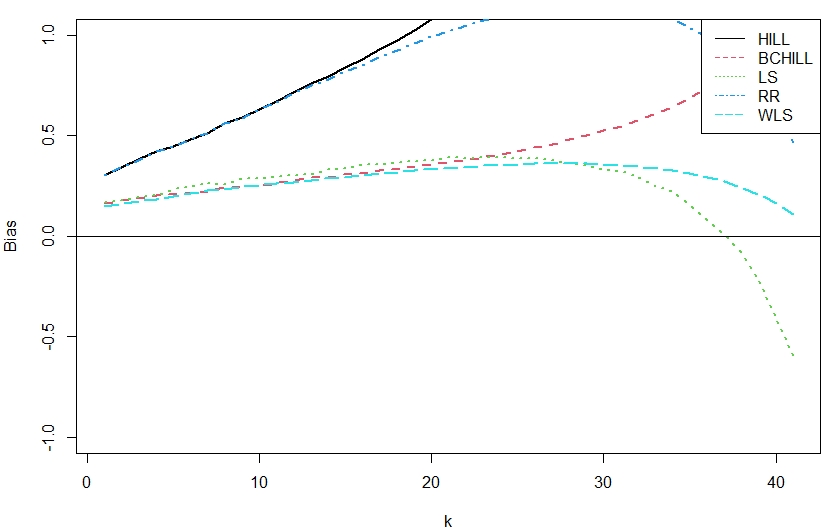}}\hfill
	\subfloat{\includegraphics[width=5.5cm,height=4cm]{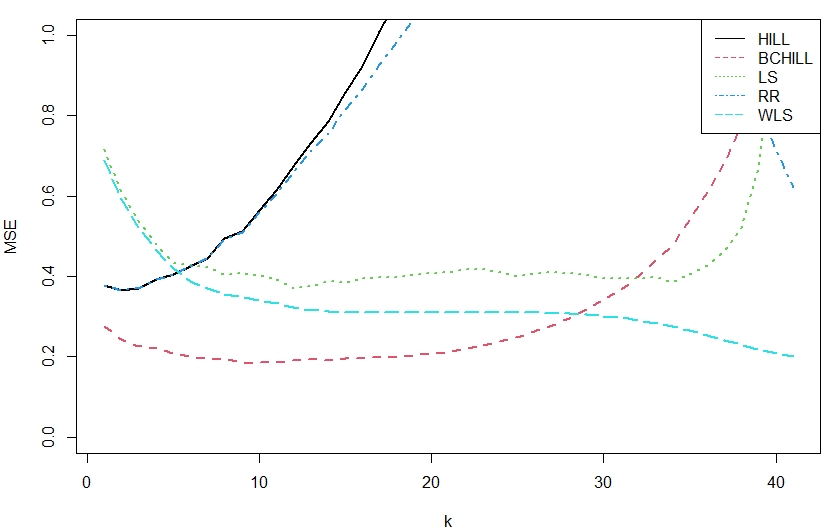}}\hfill
	\\[\smallskipamount]
	\subfloat{\includegraphics[width=5.5cm,height=4cm]{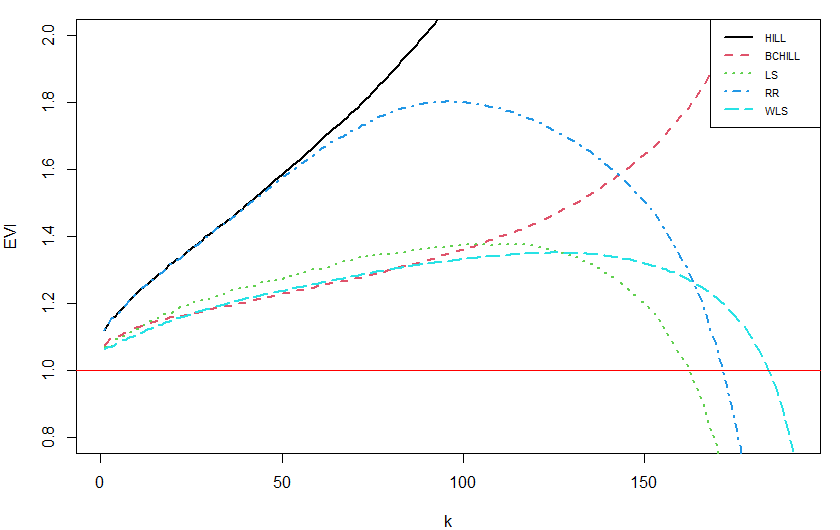}}\hfill
	\subfloat{\includegraphics[width=5.5cm,height=4cm]{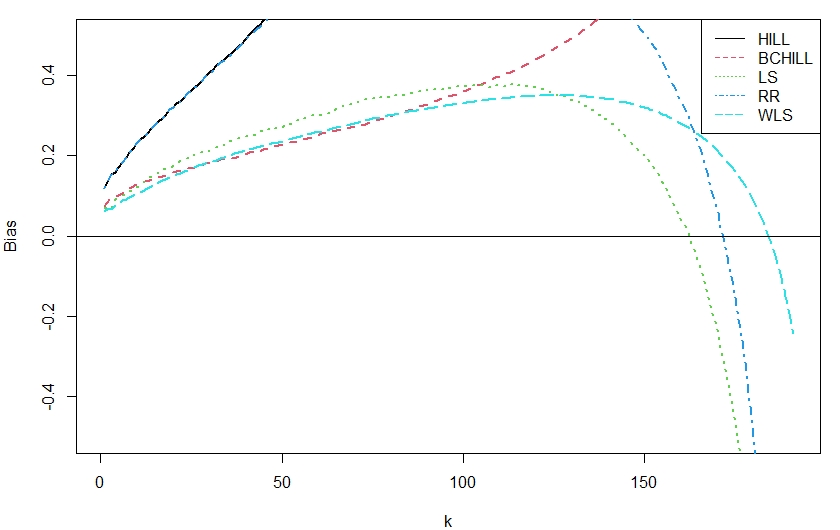}}\hfill
	\subfloat{\includegraphics[width=5.5cm,height=4cm]{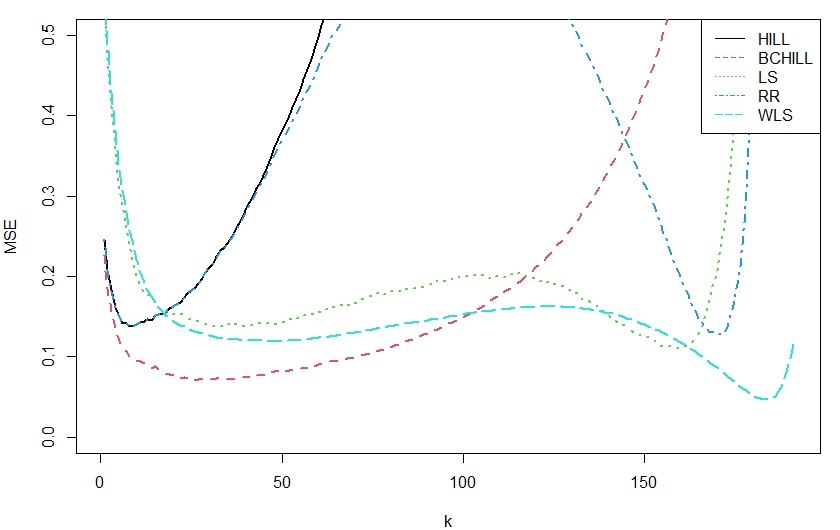}}\hfill
	\caption{Results for Burr distribution with $\gamma=1.0:$ average EVI(leftmost panel); Bias(middle panel); MSE(rightmost panel). First row: $n=50;$ second row: $n=200.$}
	\label{Fig:Figure 3}
\end{figure}

\newpage
\begin{figure}[!htbp]
	\centering
	\subfloat{\includegraphics[width=5.5cm,height=4cm]{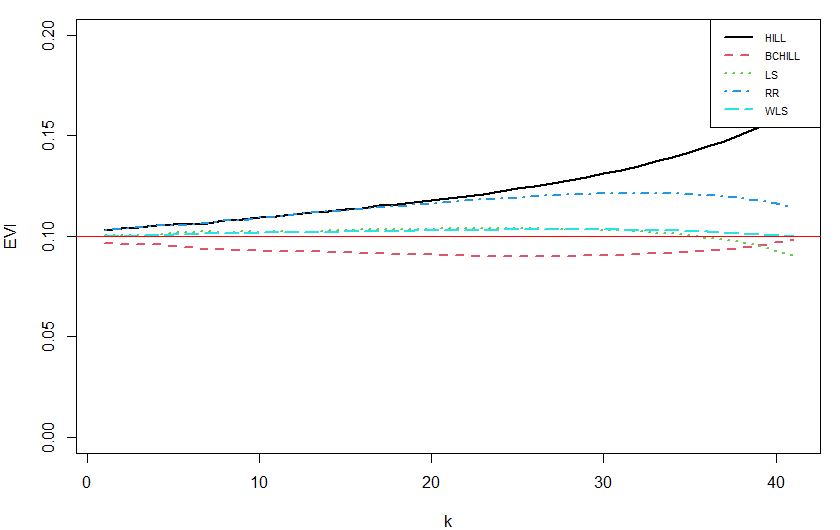}}\hfill
	\subfloat{\includegraphics[width=5.5cm,height=4cm]{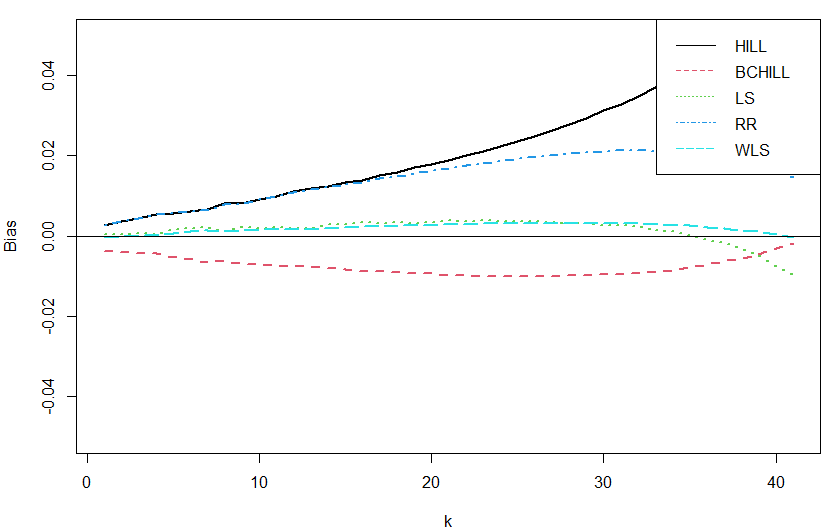}}\hfill
	\subfloat{\includegraphics[width=5.5cm,height=4cm]{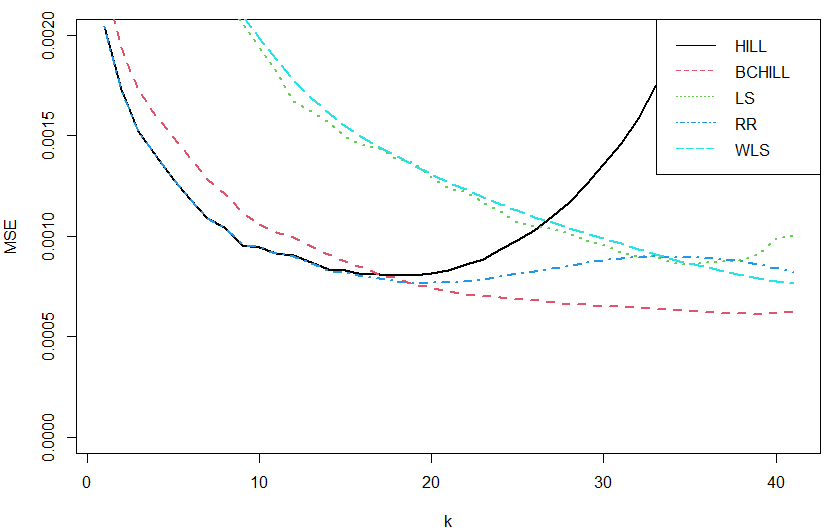}}\hfill
	\\[\smallskipamount]
	\subfloat{\includegraphics[width=5.5cm,height=4cm]{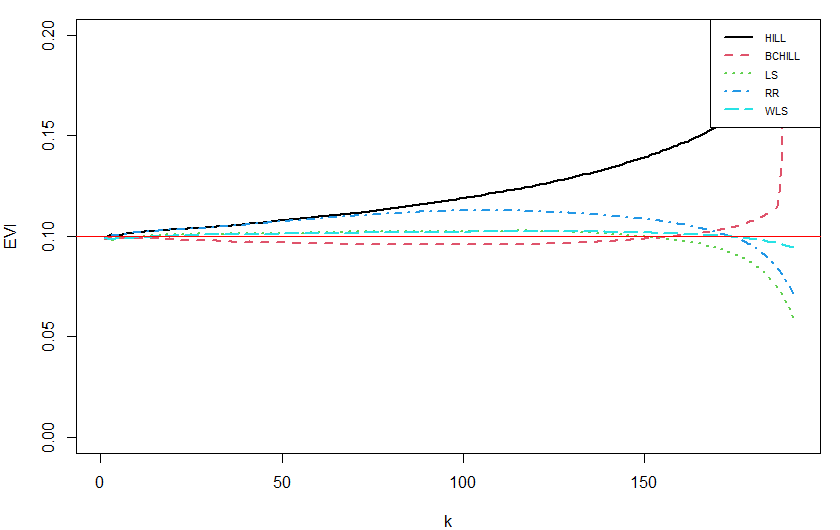}}\hfill
	\subfloat{\includegraphics[width=5.5cm,height=4cm]{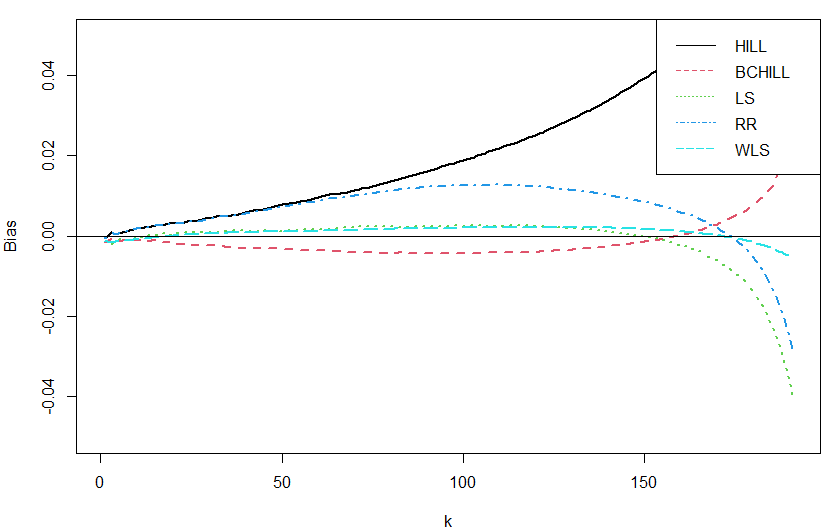}}\hfill
	\subfloat{\includegraphics[width=5.5cm,height=4cm]{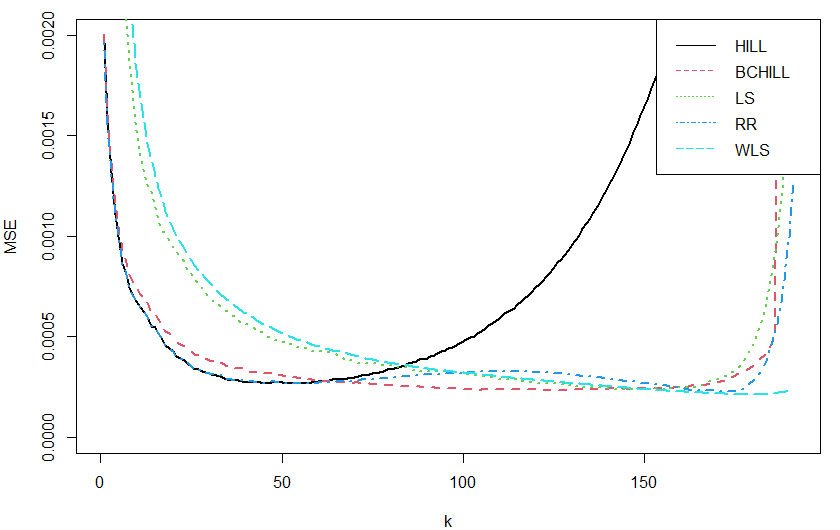}}\hfill
	\\[\smallskipamount]
	\caption{Results for Fr\'{e}chet distribution with $\gamma=0.1$: average EVI(leftmost panel); Bias(middle panel); MSE(rightmost panel). First row: $n=50$; second row: $n=200$}
	\label{Fig:Figure 4}
\end{figure}
\begin{figure}[!htbp]
	\centering
	\subfloat{\includegraphics[width=5.5cm,height=4cm]{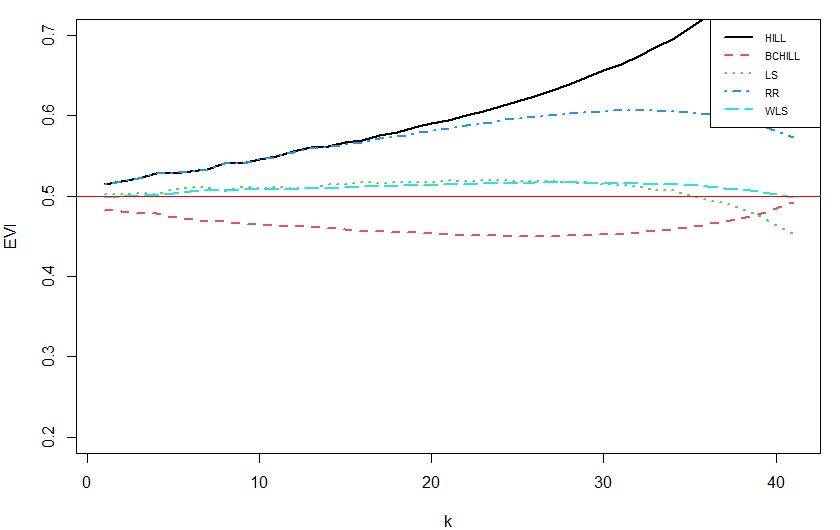}}\hfill
	\subfloat{\includegraphics[width=5.5cm,height=4cm]{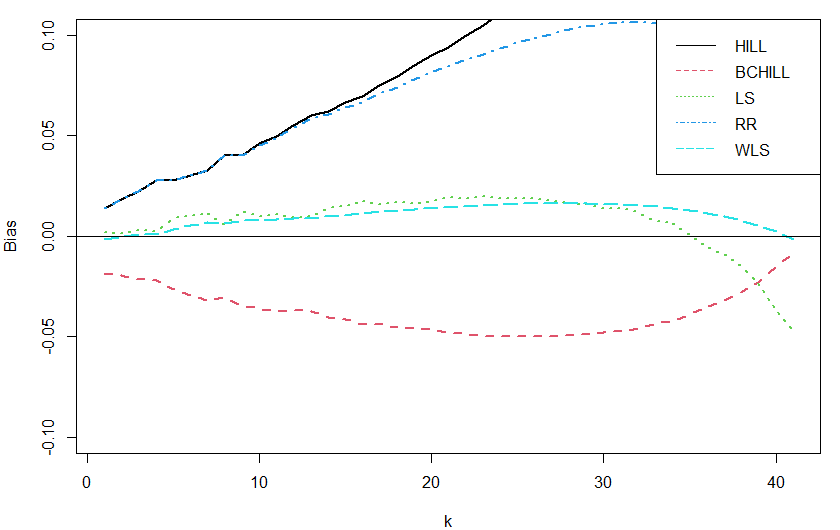}}\hfill
	\subfloat{\includegraphics[width=5.5cm,height=4cm]{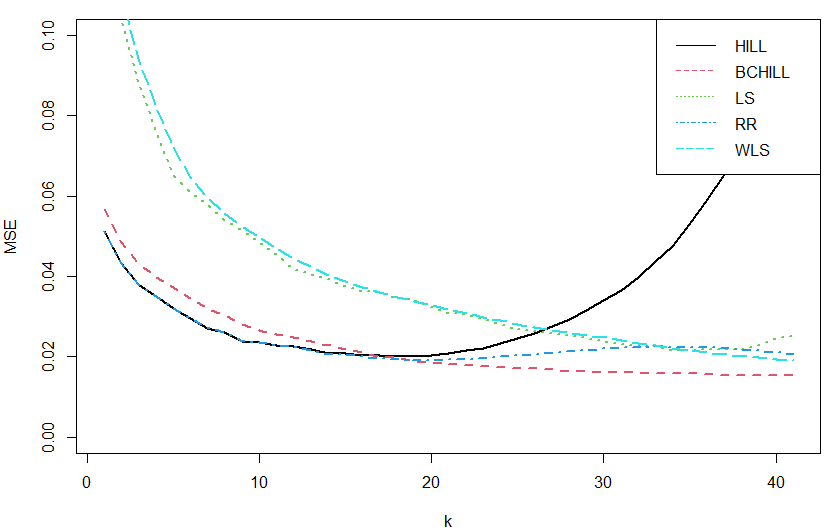}}\hfill
	\\[\smallskipamount]
	\subfloat{\includegraphics[width=5.5cm,height=4cm]{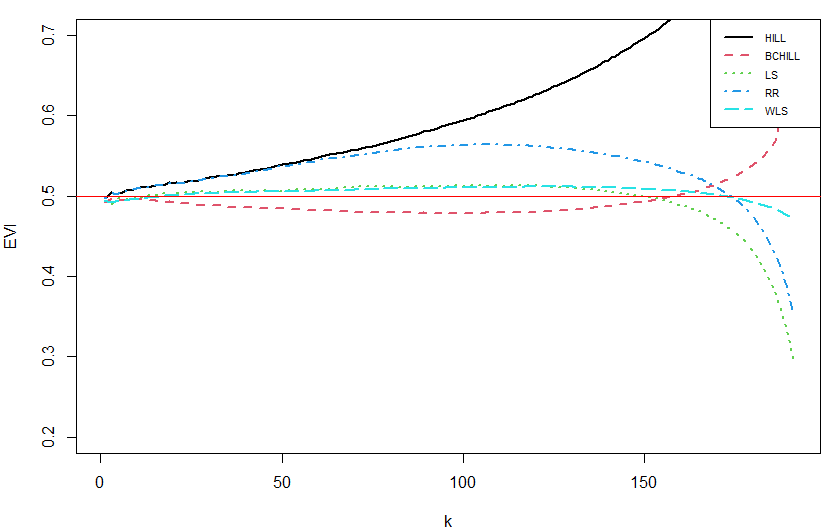}}\hfill
	\subfloat{\includegraphics[width=5.5cm,height=4cm]{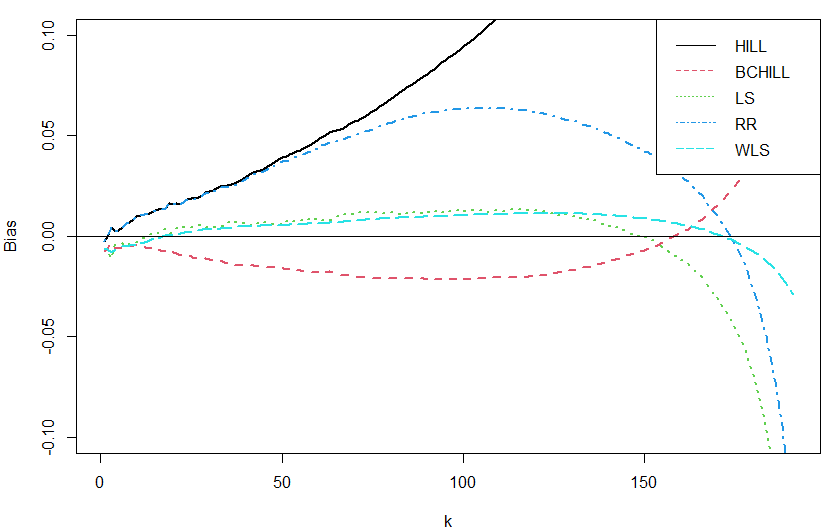}}\hfill
	\subfloat{\includegraphics[width=5.5cm,height=4cm]{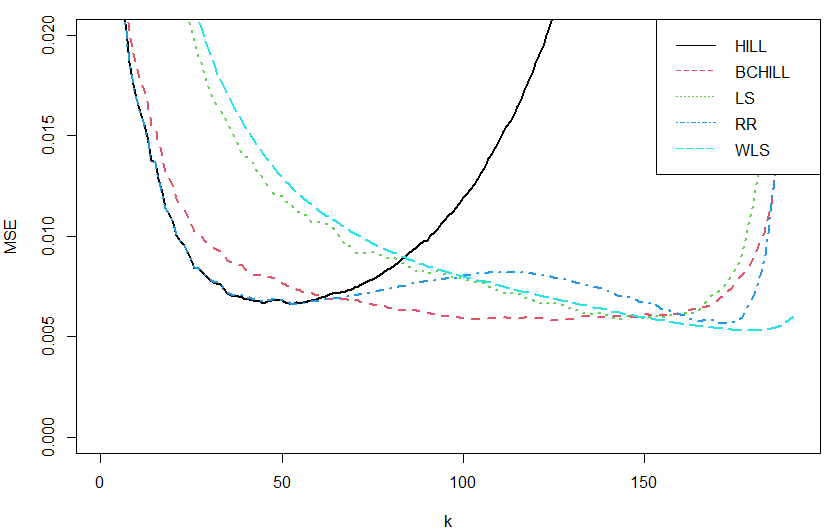}}\hfill
	\caption{Results for Fr\'{e}chet distribution with $\gamma=0.5$: average EVI(leftmost panel); Bias(middle panel); MSE(rightmost panel). First row: $n=50$; second row: $n=200.$}
	\label{Fig:Figure 5}
\end{figure}
\pagebreak
\begin{figure}[!htbp]
	\centering
	\subfloat{\includegraphics[width=5.5cm,height=4cm]{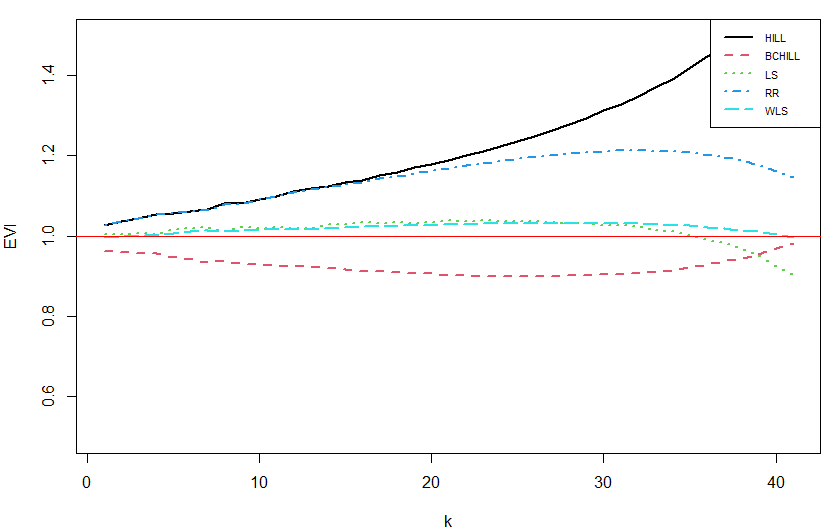}}\hfill
	\subfloat{\includegraphics[width=5.5cm,height=4cm]{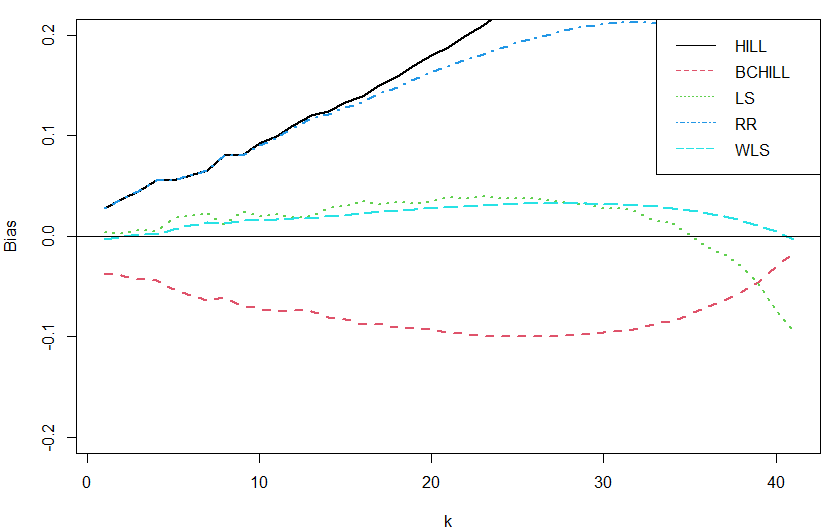}}\hfill
	\subfloat{\includegraphics[width=5.5cm,height=4cm]{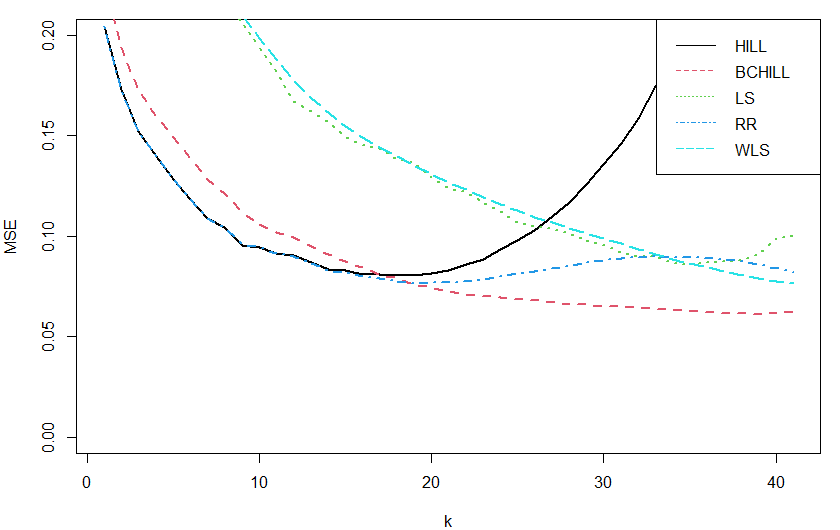}}\hfill
	\\[\smallskipamount]
	\subfloat{\includegraphics[width=5.5cm,height=4cm]{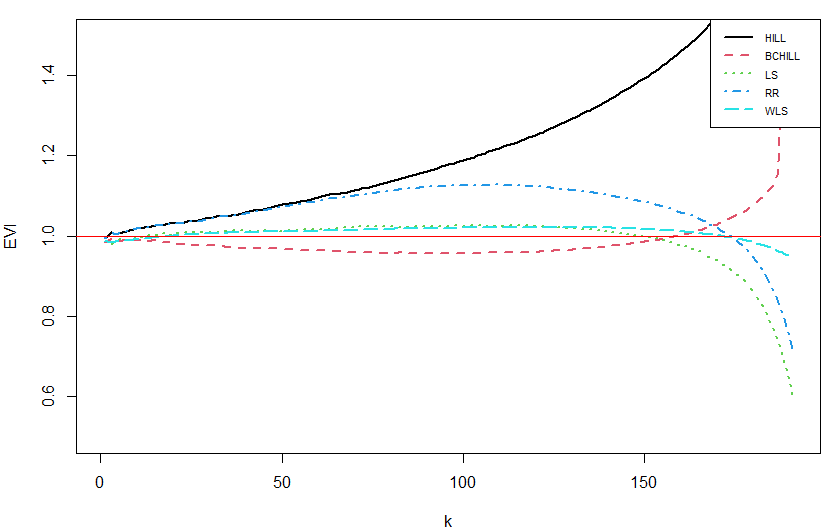}}\hfill
	\subfloat{\includegraphics[width=5.5cm,height=4cm]{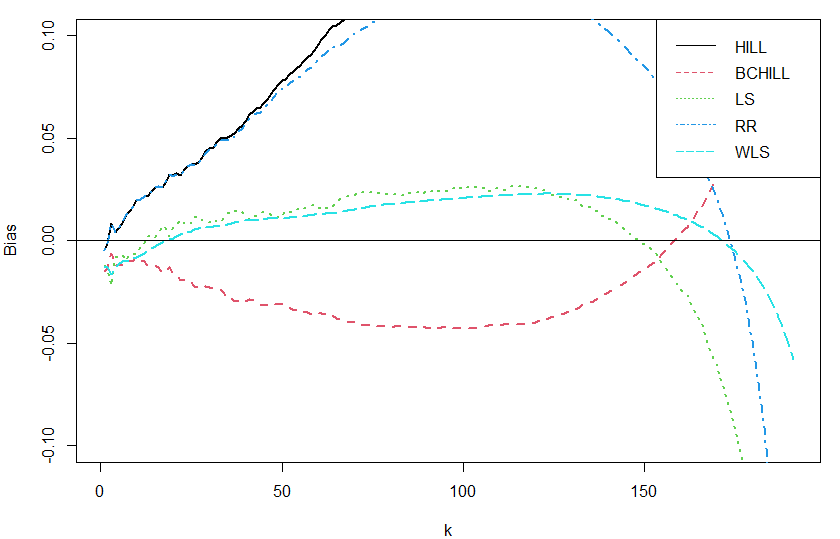}}\hfill
	\subfloat{\includegraphics[width=5.5cm,height=4cm]{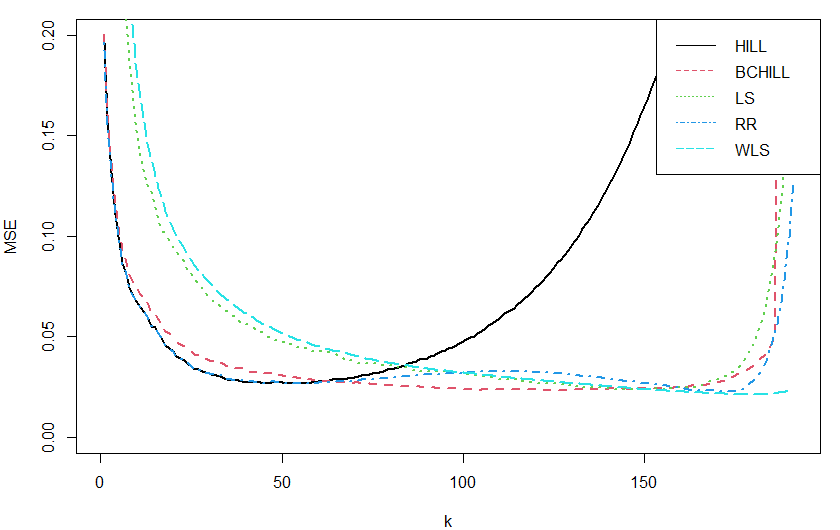}}\hfill
	\caption{Results for Fr\'{e}chet distribution with $\gamma=1.0:$ average EVI(leftmost); Bias(middle panel); MSE(rightmost panel). First row: $n=50;$ second row: $n=200.$}
	\label{Fig:Figure 6}
\end{figure}

\subsection{Practical Illustration}
To illustrate the practical application of the proposed estimator, we consider the estimation of the extreme value index of the underlying distribution of the Secura Belgian reinsurance claim size data and the calcium content (measured in mg/100g of dry soil) of soil from a particular city (NIS code 61072) in the Condroz region of Belgium. These data sets have been studied extensively in the extreme value theory literature such as \citet{Beirlant2004}. In this application, we lookout for estimators whose graph is near horizontal and smooth.

The Secura Belgian Reinsurance data contains 371 automobile claims in euros which occurred from 1989 to 2001.  \citet{Beirlant2004} have demonstrated that the Secura Belgian reinsurance data has a heavy tail. The plots of the extreme value index estimates as a function $k$ is shown in Figure~\ref{Fig:Figure 6}.

Also, the Condroz data contains 1505 observations, and seven of these values deviate from the rest of the observations.  \citet{Beirlant2004} removed the top seven Calcium measurements before modelling the data. However, we fitted the biased-reduced weighted least square to the complete data to study the actual tail behaviour of the data. \citet{Beirlant2004} have shown that the Ca measurements in the Condroz data belongs to the Pareto domain of attraction. The extreme value index estimates of the Ca measurements of the Condroz data is illustrated in Figure~\ref{Fig:Figure 7}.

In the case of the Secura Belgian data, the HILL estimator is very sensitive to the choice of $k$, and this makes its interpretation challenging. The BCHILL and LS estimators exhibit some moderate stability in $k$. The ridge regression estimates are stable for $k$ values between $110$ and $220$, while the reduced-bias weighted least square yields stable estimates for $k\ge 120$. The WLS shows stability across a larger region of $k$, making it easier to specify the value of the estimated extreme value index. 

Regarding the Condroz data, the plots of  HILL and BCHILL estimators are difficult to interpret for large values of $k$, and this is because the estimators are very sensitive to the choice of $k$. However, they are less sensitive for $k$ values between $270$ and $500$. The extreme value index estimates for the RR, LS and WLS estimators are stable and approximately equal to $0.26$ for $k$ values in the interval $[283,1097]$, $[505, 1107]$ and $[710,1230]$ respectively. 

The $\rho$ estimator by \citet{alves2003estimation} was also considered for the application but the result is omitted. The reduced-bias estimators produced stable estimates over some $k$ values but not as horizontal as when the minimum variance \citep{Buitendag2018} approach is used, especially in the case of the Secura Belgian claim data whose sample size is relatively small.

\begin{figure}[!htbp]
	\centering
	\subfloat{\includegraphics[width=8cm,height=5cm]{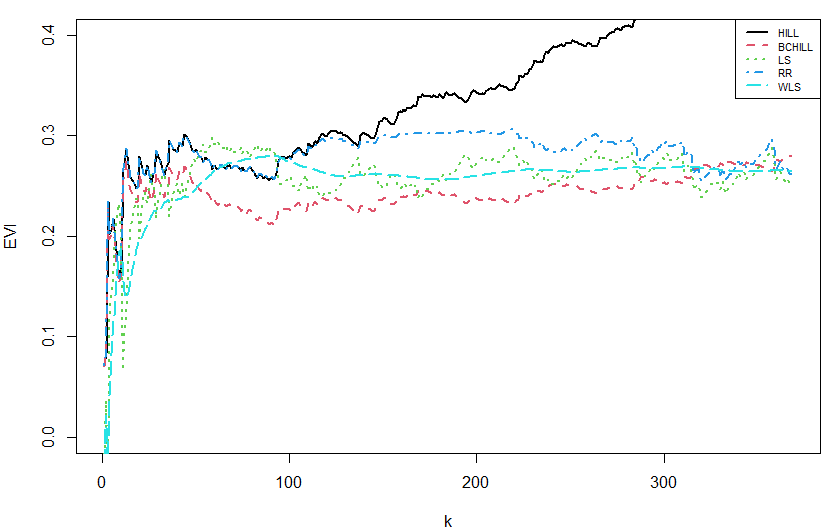}}\hfill
	\subfloat{\includegraphics[width=8cm,height=5cm]{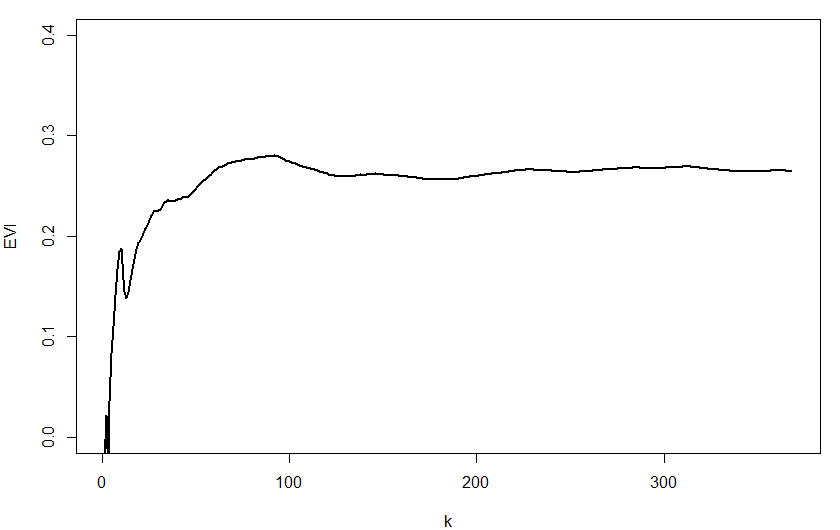}}\hfill
	\caption{Claim sizes of Secura Belgian Reinsurance Data: $\hat{\gamma}_{wls}^{+}$ on the right}
	\label{Fig:Figure 7}
\end{figure}
\begin{figure}[!htbp]
	\centering
	\subfloat{\includegraphics[width=8cm,height=5cm]{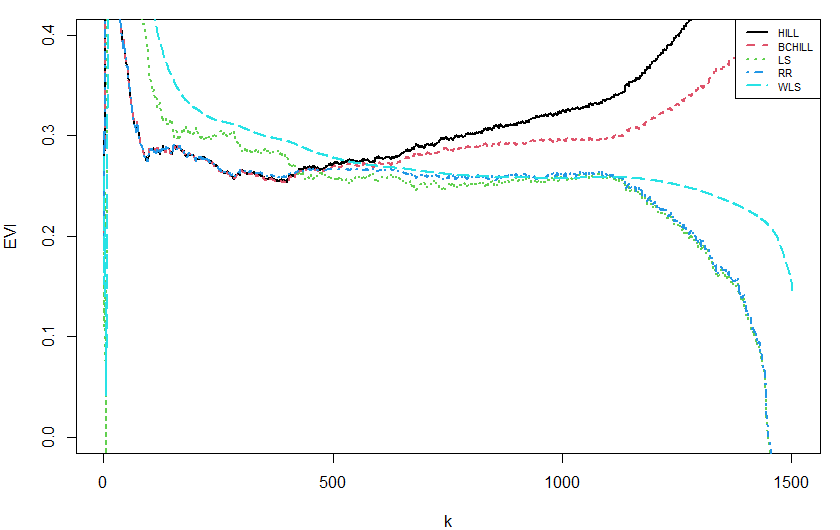}}\hfill
	\subfloat{\includegraphics[width=8cm,height=5cm]{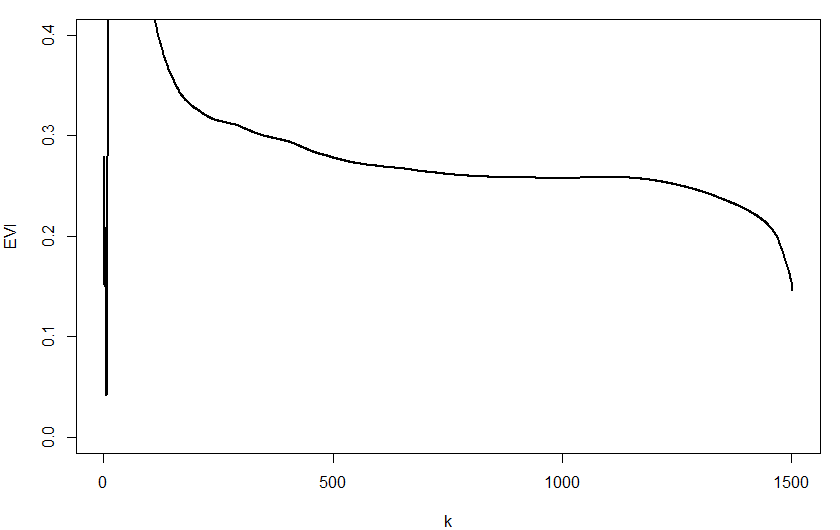}}\hfill
	\caption{Calcium content of soil. Condroz Data: $\hat{\gamma}_{wls}^{+}$ on the right}
	\label{Fig:Figure 8}
\end{figure}

\section{\textbf{Conclusion}}\label{S4}
This study proposed an alternative reduced-bias estimator for the extreme value index, $\gamma (\gamma>0),$ using the exponential regression model introduced by \citet{Beirlant2002}. Specifically, we proposed a reduced-bias weighted least squares estimator of the extreme value index using the exponential regression model for the weighted log-spacings. We showed that the proposed estimator is asymptotically unbiased, asymptotically consistent and asymptotically normal. In addition, it performed favourably well in terms of bias and MSE relative to the other EVI estimators considered in the study. Furthermore, it produces stable estimates that are less sensitive to the tail sample fraction, $k$. Hence it can easily be used in the selection of the appropriate value of $k$ which is an ongoing research area in extreme value analysis. In application, we suggest practitioners plots $\hat{\gamma}_{wls}^{+}$ as a function of $k$ to aid in the selection of an optimal $k$ for the proposed estimator.


Unlike the least squares and the ridge regression estimators in the literature, the asymptotic variance of the proposed estimator is $\sigma^2=4/3\gamma^2$ and does not depend on $\rho.$ Also, the asymptotic variance of the proposed estimator is less than that of the least squares estimator for $-6\le \rho<0.$ Moreover, we envisage that if we choose a random weight, we can attain an asymptotic variance that is as smaller as that of the Hill or the bias-corrected Hill estimators. This is a topic that deserves further studies and will be considered in a subsequent work.
\subsection*{Data Availability}
The data sets used in this study to support the findings are from \citet{Beirlant2004} and are available at \url{https://lstat.kuleuven.be/Wiley/}.

\subsection*{Conflicts of Interest}

The authors declare that there is no conflict of interest regarding the publication of this paper.\\

\subsection*{Acknowledgement}

Note, a draft of  the  article  has  previously  appeared online via arxiv.org as  a  preprint. See, \cite{ocran2021reduced}.  Also, Ocran,  E. would like to thank the University of Ghana Building a New Generation of Academics in Africa (BANGA-Africa) Project (funded by Carnegie Corporation of New York )  for providing financial support  for  this Ph.D  research work.\\
\newpage
\section*{Appendix: Log-Gamma Distribution}

\begin{figure}[!htbp]
	\centering
	\subfloat{\includegraphics[width=5.5cm,height=4cm]{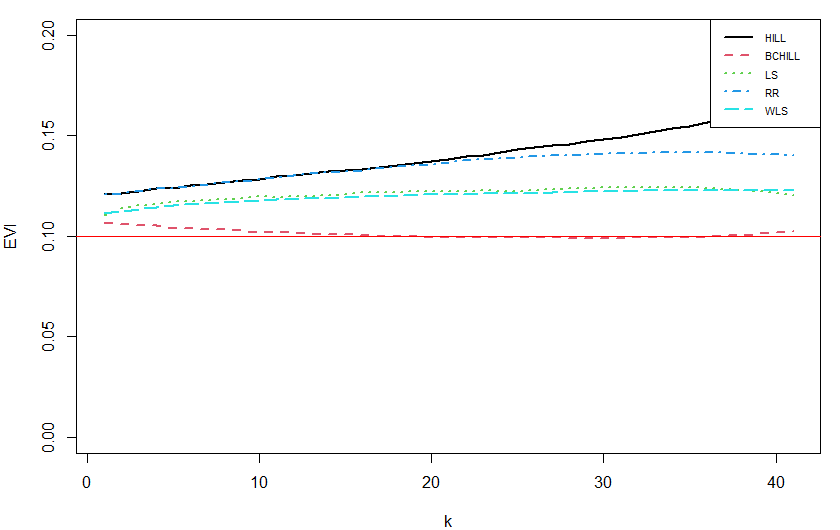}}\hfill
	\subfloat{\includegraphics[width=5.5cm,height=4cm]{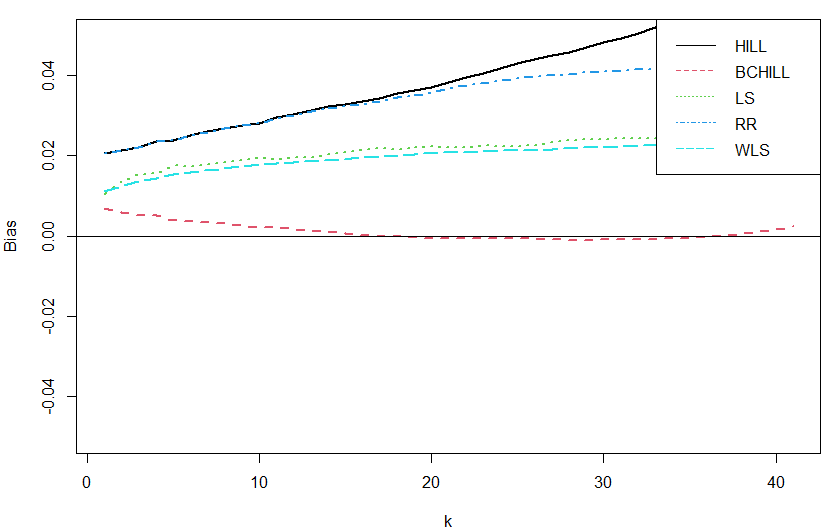}}\hfill
	\subfloat{\includegraphics[width=5.5cm,height=4cm]{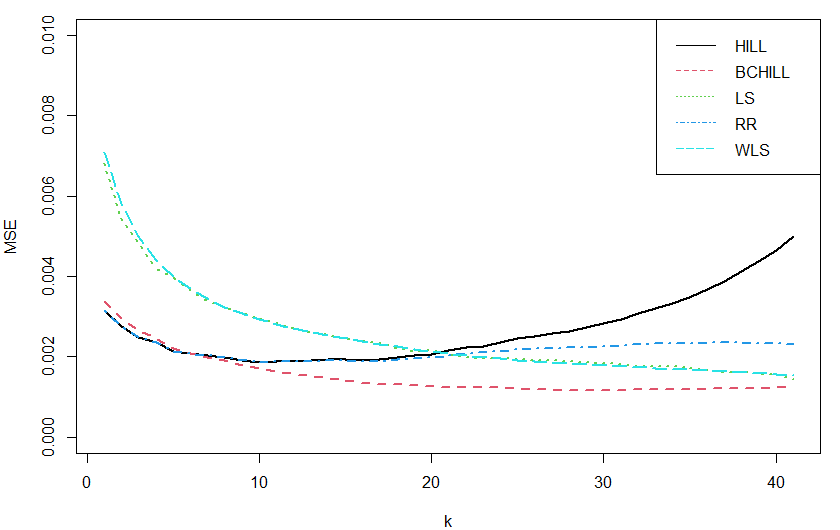}}\hfill
	\\[\smallskipamount]
	\subfloat{\includegraphics[width=5.5cm,height=4cm]{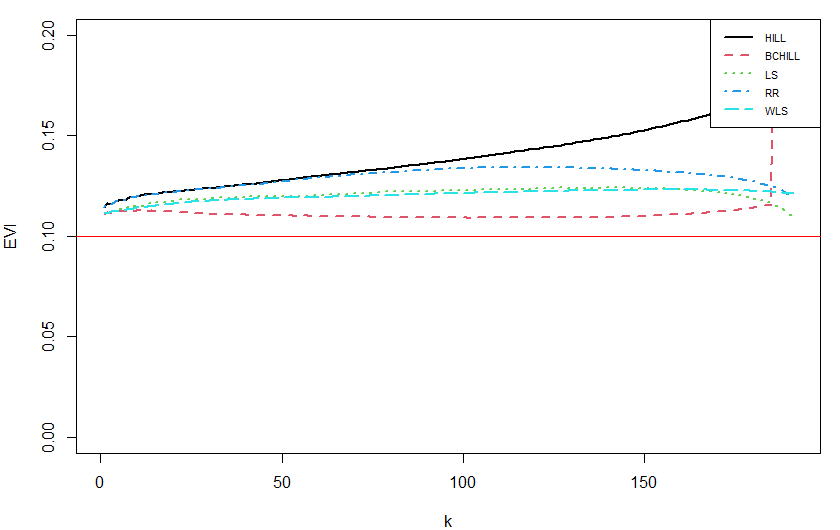}}\hfill
	\subfloat{\includegraphics[width=5.5cm,height=4cm]{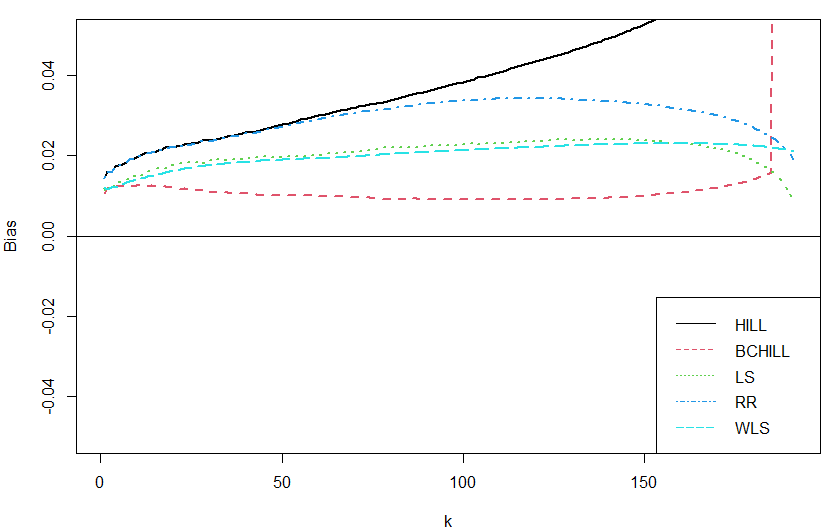}}\hfill
	\subfloat{\includegraphics[width=5.5cm,height=4cm]{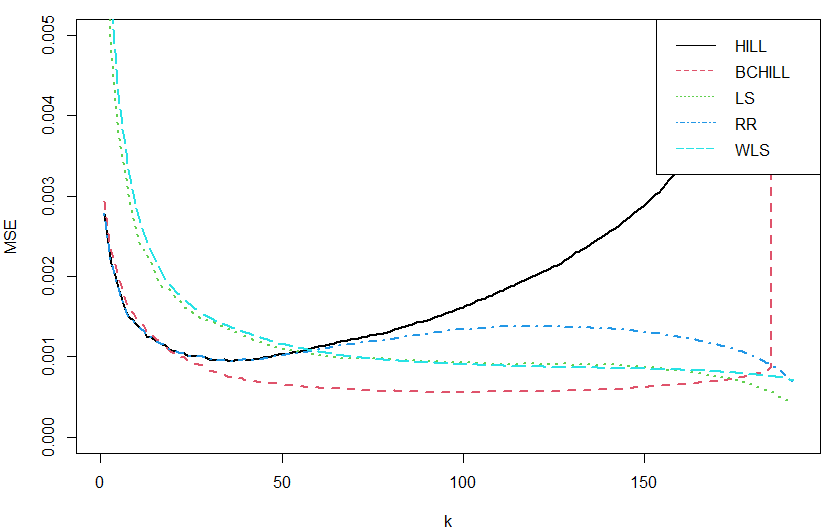}}\hfill
	\\[\smallskipamount]
	\caption{Results for Log-Gamma  distribution with $\gamma=0.1$: average EVI(leftmost panel); Bias(middle panel); MSE(rightmost panel). First row: $n=50$; second row: $n=200$}
	\label{Fig:Figure A1}
\end{figure}
\begin{figure}[!htbp]
	\centering
	\subfloat{\includegraphics[width=5.5cm,height=4cm]{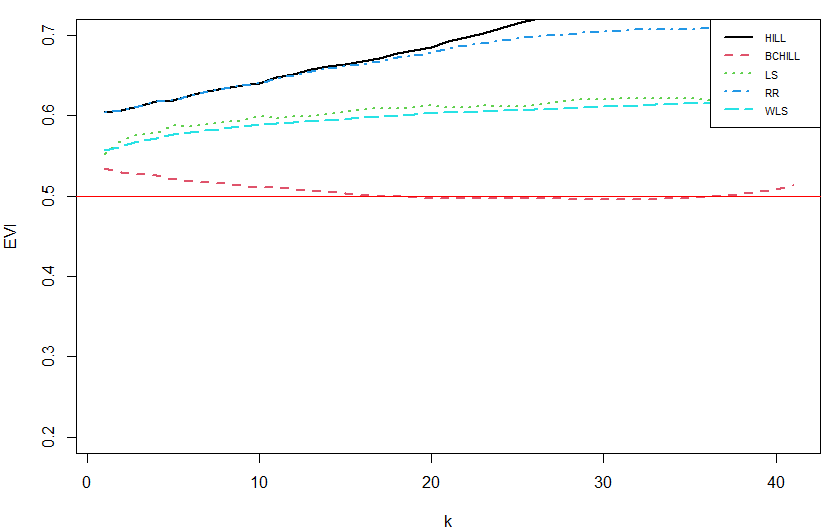}}\hfill
	\subfloat{\includegraphics[width=5.5cm,height=4cm]{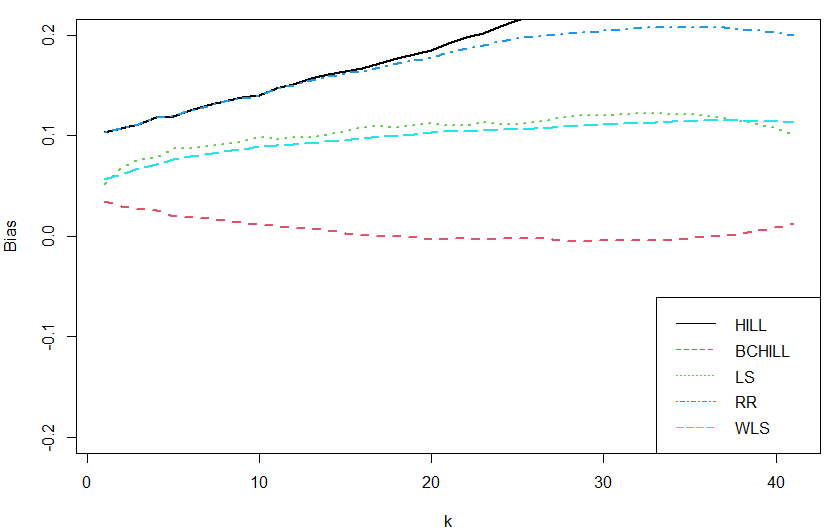}}\hfill
	\subfloat{\includegraphics[width=5.5cm,height=4cm]{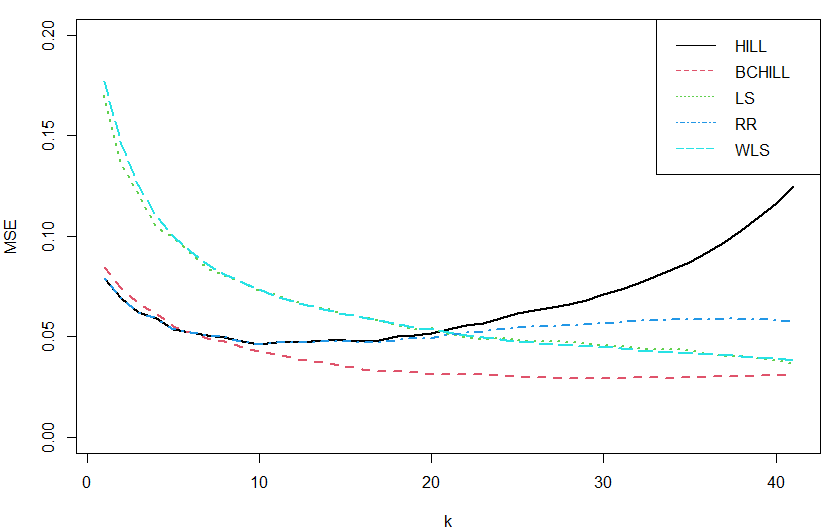}}\hfill
	\\[\smallskipamount]
	\subfloat{\includegraphics[width=5.5cm,height=4cm]{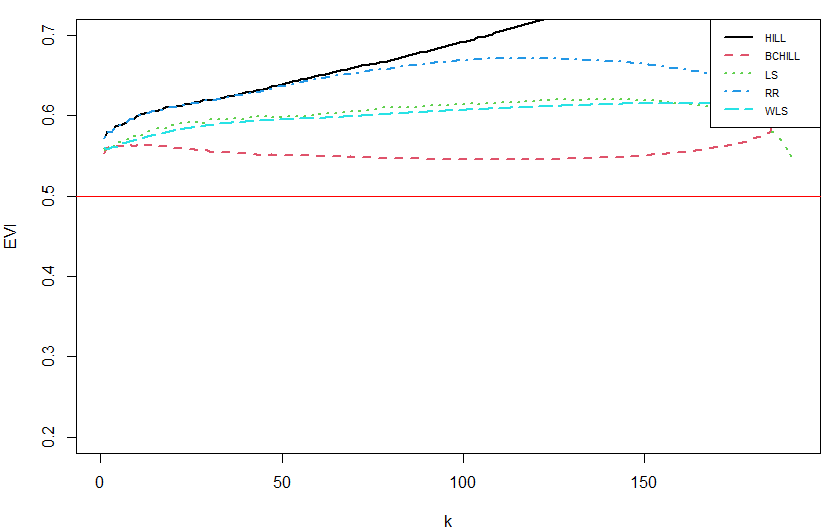}}\hfill
	\subfloat{\includegraphics[width=5.5cm,height=4cm]{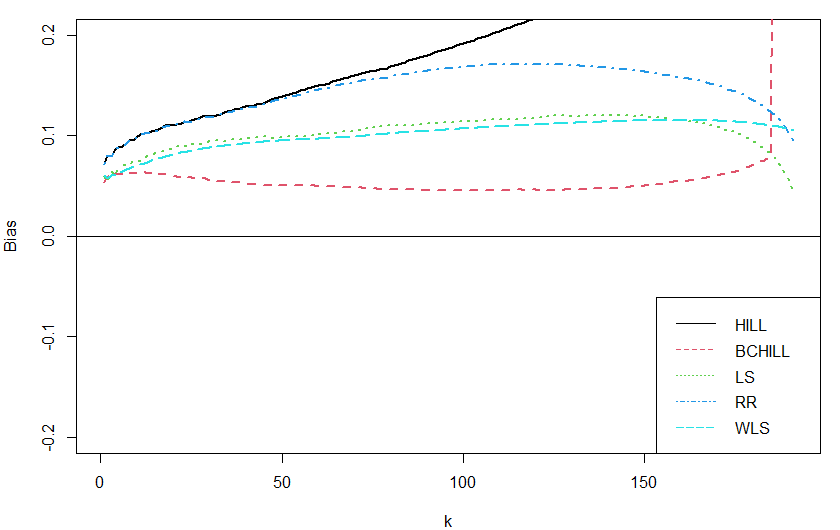}}\hfill
	\subfloat{\includegraphics[width=5.5cm,height=4cm]{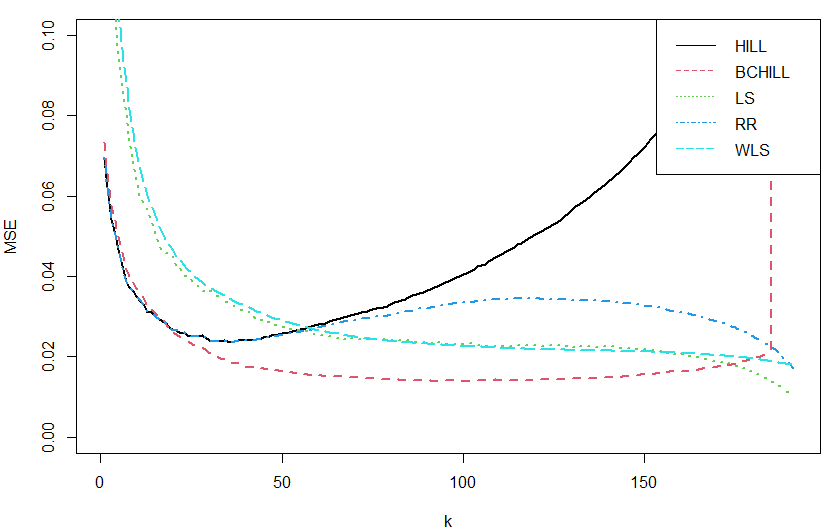}}\hfill
	\caption{Results for Log-Gamma distribution with $\gamma=0.5$: average EVI(leftmost panel); Bias(middle panel); MSE(rightmost panel). First row: $n=50$; second row: $n=200.$}
	\label{Fig:Figure A2}
\end{figure}
\begin{figure}[!htbp]
	\centering
	\subfloat{\includegraphics[width=5.5cm,height=4cm]{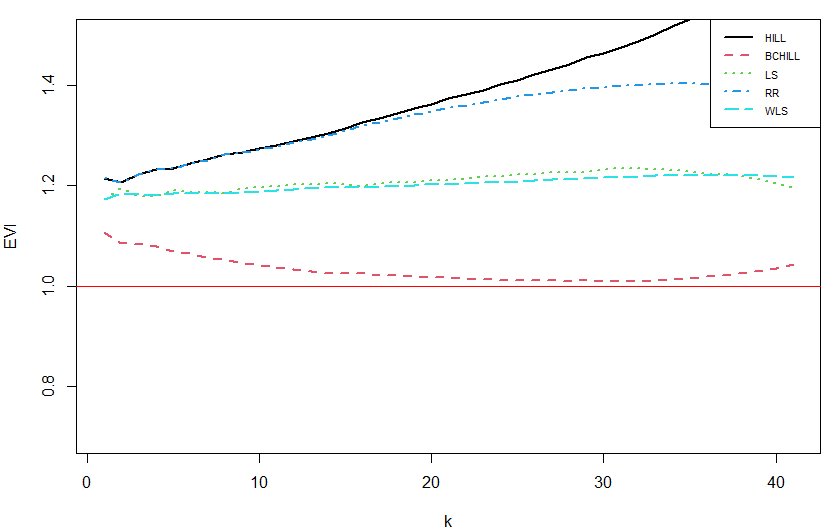}}\hfill
	\subfloat{\includegraphics[width=5.5cm,height=4cm]{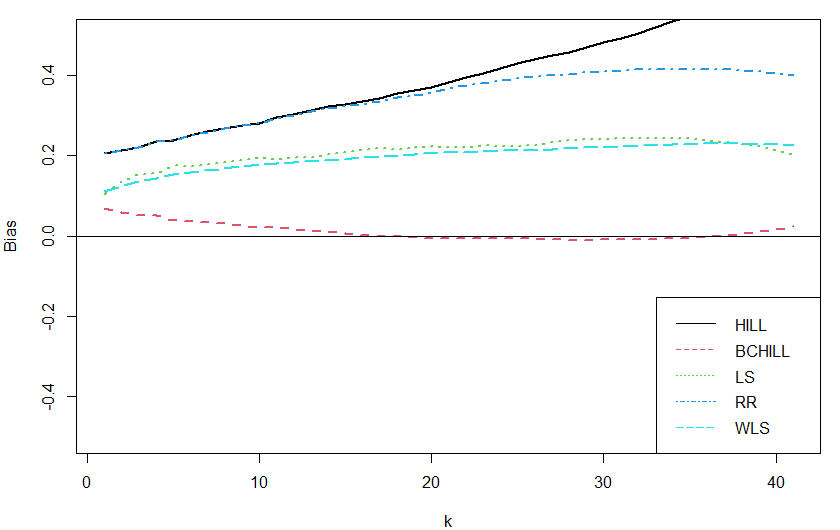}}\hfill
	\subfloat{\includegraphics[width=5.5cm,height=4cm]{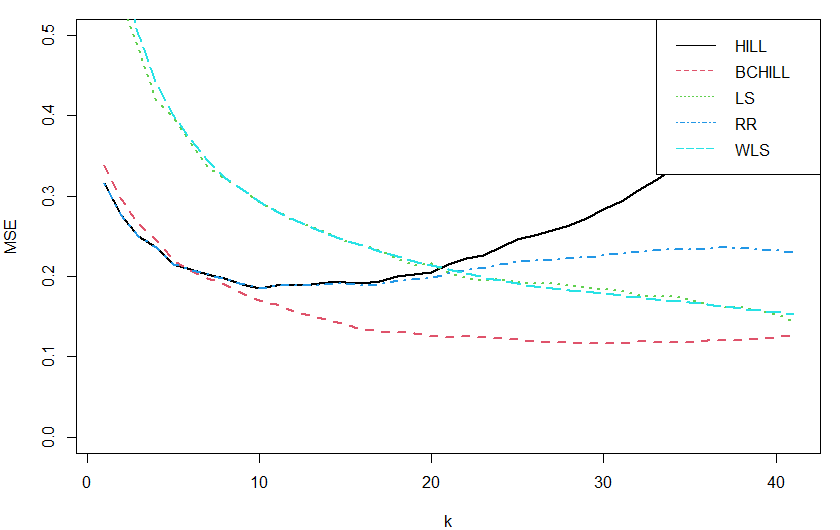}}\hfill
	\\[\smallskipamount]
	\subfloat{\includegraphics[width=5.5cm,height=4cm]{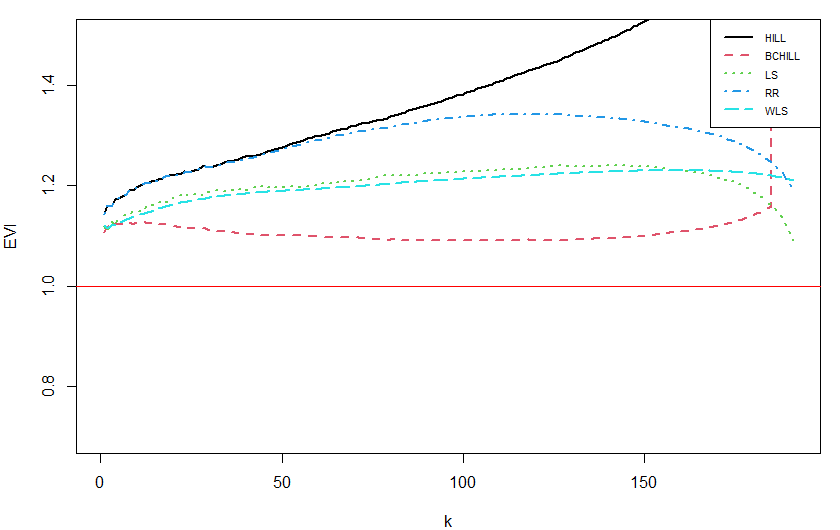}}\hfill
	\subfloat{\includegraphics[width=5.5cm,height=4cm]{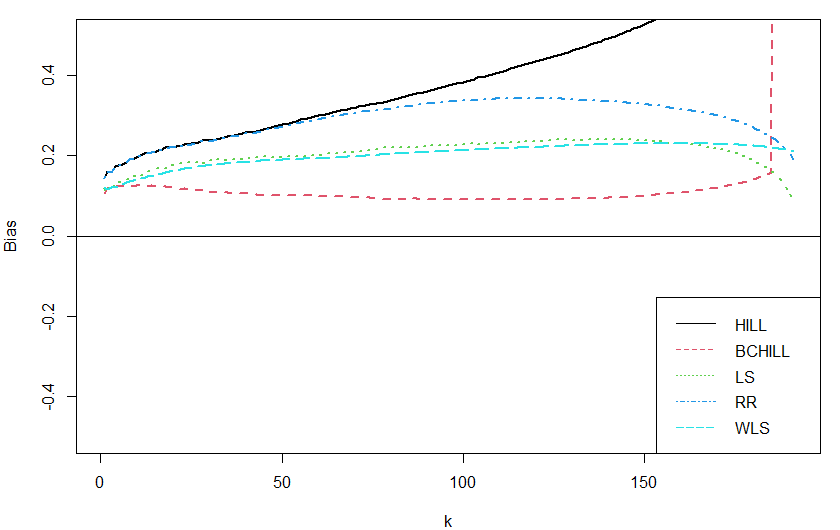}}\hfill
	\subfloat{\includegraphics[width=5.5cm,height=4cm]{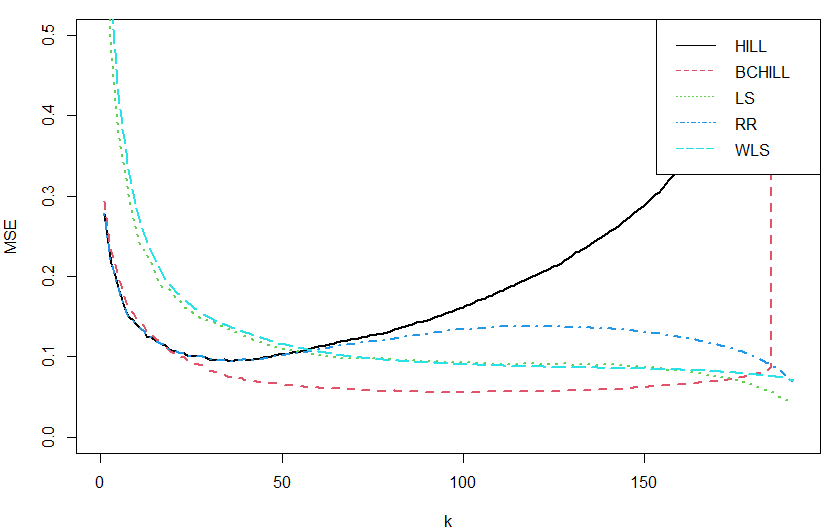}}\hfill
	\caption{Results for Log-Gamma distribution with $\gamma=1.0:$ average EVI(leftmost); Bias(middle panel); MSE(rightmost panel). First row: $n=50;$ second row: $n=200.$}
	\label{Fig:Figure A3}
\end{figure}

\bibliographystyle{apalike}

\end{document}